%% file: main.tex
\documentclass[journal,twoside,web]{ieeecolor}
\usepackage{generic}
\usepackage{cite}
\usepackage{amsmath,amssymb,amsfonts,amsthm}
\usepackage{algorithmic}
\usepackage{graphicx}
\usepackage{textcomp}
\usepackage{subcaption}
\usepackage{mathtools}
\usepackage[acronym]{glossaries}
\usepackage{cleveref}
\usepackage{xcolor}
\usepackage{array}
\usepackage{tikz}

\newcommand{\Vspace}{\vphantom{\vdots}}

\let\labelindent\relax
\usepackage[shortlabels]{enumitem}

\newcommand\undermat[2]{%
  \makebox[0pt][l]{$\smash{\underbrace{\phantom{%
    \begin{matrix}#2\end{matrix}}}_{\text{$#1$}}}$}#2}

\DeclareMathOperator*{\argmin}{arg\,min}

\newtheorem{theorem}{Theorem}
\newtheorem{lemma}{Lemma}

\newtheorem{remark}{Remark}
\newtheorem{corollary}{Corollary}
\newtheorem{definition}{Definition}

\newtheorem{assumption}{Assumption}

\newcommand{\ak}[1]{{\color{black}{#1}}}

\newacronym{ediss}{\emph{E-$\delta$-ISS}}{exponentially incrementally input-to-state stable}
\newacronym{mpc}{MPC}{model predictive control}
\newacronym{roa}{ROA}{region of attraction}
\newacronym{es}{ES}{exponential stability}
\newacronym{lqmpc}{LQMPC}{linear quadratic MPC}

\newcommand\copyrighttext{%
  \footnotesize \copyright 2025 IEEE. Personal use of this material is permitted. Permission from IEEE must be obtained for all other uses, in any current or future media, including reprinting/republishing this material for advertising or promotional purposes, creating new collective works, for resale or redistribution to servers or lists, or reuse of any copyrighted component of this work in other works.}
\newcommand\copyrightnotice{%
\begin{tikzpicture}[remember picture,overlay]
\node[anchor=south,yshift=2pt] at (current page.south) {\fbox{\parbox{\dimexpr\textwidth-\fboxsep-\fboxrule\relax}{\copyrighttext}}};
\end{tikzpicture}%
}

\newcommand\publishedtext{%
  \footnotesize \copyright Published in IEEE Transactions on Automatic Control\\
  DOI: 10.1109/TAC.2025.3539988 }
\newcommand\publishednotice{%
\begin{tikzpicture}[remember picture,overlay]
\node[anchor=north,yshift=-3pt] at (current page.north) {\fbox{\parbox{\dimexpr\textwidth-\fboxsep-\fboxrule\relax}{\publishedtext}}};
\end{tikzpicture}%
}

\begin{document}
\title{Closed-Loop Finite-Time Analysis of \\Suboptimal Online Control}
\author{Aren Karapetyan, Efe C. Balta, Andrea Iannelli, and John Lygeros
\thanks{This work has been supported by the Swiss National Science Foundation under NCCR Automation (grant agreement  $51\text{NF}40\_225155$) and by the German Research Foundation (DFG) under Germany’s Excellence Strategy - EXC 2075 – 390740016.}%
\thanks{A. Karapetyan, and J. Lygeros are with the Automatic Control Laboratory, Swiss Federal Institute of Technology (ETH Z\"urich), 8092 Z\"urich, Switzerland
       (E-mails: {\tt\small \{akarapetyan, lygeros\}@control.ee.ethz.ch}).}%
\thanks{E. C. Balta  is  with  Inspire AG, 8005 Zürich, Switzerland (E-mail: {\tt\small efe.balta@inspire.ch}).}
\thanks{A. Iannelli is with the Institute for Systems Theory and Automatic Control, University of Stuttgart,  Stuttgart 70569,  Germany
        (E-mail: {\tt\small andrea.iannelli@ist.uni-stuttgart.de}).}%
}

\maketitle

\publishednotice

\copyrightnotice

\begin{abstract}
Suboptimal methods in optimal control arise due to a limited computational budget, unknown system dynamics, or a short prediction window among other reasons. \ak{In this work, we study the transient closed-loop performance of such methods by providing finite-time suboptimality gap guarantees.  We consider the control of discrete-time, nonlinear time-varying dynamical systems and establish sufficient conditions for such guarantees.} These allow the control design to distribute a limited computational budget over a time horizon and estimate the on-the-go loss in performance due to suboptimality. We study exponential incremental input-to-state stabilizing policies and show that for nonlinear systems, under some mild conditions, this property is directly implied by exponential stability without further assumptions on global smoothness. The analysis is showcased on a suboptimal model predictive control use case.

\end{abstract}


\input{S1_introduction.tex}
\input{S2_preliminaries.tex}

\input{S3_suboptimality_gap.tex}
\input{S4_incremental_stability.tex}
\input{S5_usecases.tex}

\input{S6_conclusions.tex}
\input{SA_appendix.tex}
\bibliographystyle{ieeetr}
\bibliography{bibliography.bib}

\end{document}

%% file: S1_introduction.tex
\section{Introduction} \label{sec:introduction}

Optimal control aims to compute an input signal to drive a dynamical system to a given target state, while optimizing a performance cost subject to constraints. In the absence of uncertainty, the problem has been studied using calculus of variations \cite{pontryagin1987mathematical} and dynamic programming \cite{bellman1966dynamic}. However, in many practical applications with limited computational power, it becomes difficult or infeasible to solve due to the curse of dimensionality  \cite{bellman1966dynamic}. This is further exacerbated if there are unknown system and/or cost parameters. As a result, control designers rely on approximate or suboptimal methods \cite{bertsekas2022lessons} to solve the problem. If there are adequate computational resources and an accurate simulator of the true system, the problem can be solved up to an arbitrary accuracy using approximate dynamic programming \cite{bertsekas2012dynamic} or reinforcement learning \cite{sutton2018reinforcement} techniques. When this is not the case, e.g. the system has unpredictable dynamics or the cost to be optimized for is changing adversarially, offline methods alone are not sufficient. In such cases, the policy is updated online or adaptively as more data becomes available. We consider a  controller to be \textit{online}, rather than \textit{offline} or pre-fixed if it is capable of adapting to changes in the control problem during its single trajectory execution. For example, a model predictive controller (MPC) can utilize a finite window of predictions for possibly time-varying dynamics, disturbances, references, or costs to generate an input at each step. While such a definition partially overlaps with that of adaptive control \cite{aastrom2013adaptive}, it defines a wider class of controllers, that also includes online optimization-based methods, like \cite{li2019online,hazan2020nonstochastic}.

\begin{figure}%
    \begin{center}
    \includegraphics[width=0.88\columnwidth]{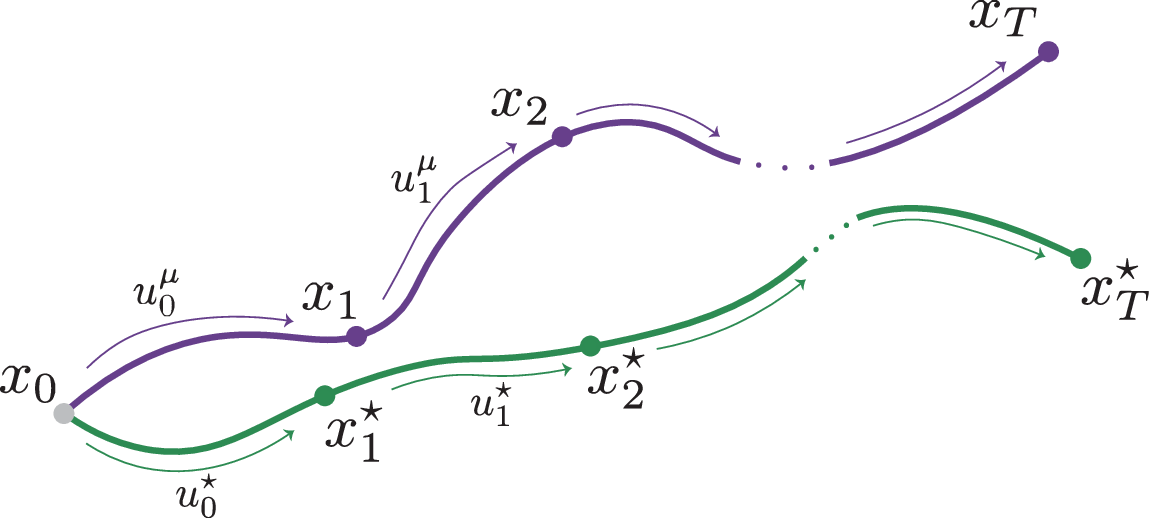}    
    \caption{Two separate closed-loop trajectories, generated by applying a suboptimal input signal $\boldsymbol{u}^\mu$, and a benchmark input signal $\boldsymbol{u^\star}$.}
    \label{fig:trajectories}
    \end{center}
\end{figure}

The focus of this work is the closed-loop suboptimality analysis of online controllers in the finite-time or transient domain. Given their real-time implementation, online methods need to stay computationally efficient while stabilizing the system. Additionally, their performance is measured in terms of the accumulated cost that needs to be kept to a minimum. To quantify this, we fix a benchmark policy that we  may deem to be close to the desired optimal one,  visualized in Figure \ref*{fig:trajectories}, and study the suboptimality gap of the given online algorithm in terms of the additional incurred cost due to its suboptimality with respect to the benchmark. Such an analysis provides a relative measure on the performance of the given algorithm, since, in general, the benchmark policy attains a non-zero cost. In this context, we pose the following questions.
\begin{enumerate}
    \item \emph{How does the transient cost performance of an online algorithm scale with a measure of its suboptimality?}
    \item \emph{How should the benchmark policy be chosen to achieve meaningful finite-time bounds?}
\end{enumerate}

We consider nonlinear time-varying systems
and derive conditions on the suboptimal and benchmark policies such that the suboptimality gap, defined as the difference of their respective closed-loop costs, can be quantified. We show that it scales with the pathlength of the suboptimal trajectory and the rate of convergence of the suboptimal policy to the benchmark. The former can be computed online providing an on-the-go estimate of the suboptimality, and the latter usually depends on suboptimal algorithm parameters, allowing the control designer to tune these accordingly. We study the suboptimal, projected gradient method (PGM)-based linear-quadratic MPC \cite{liao2021analysis} as an example of online control satisfying the proposed assumptions, and bound its suboptimality gap with respect to optimal MPC. Our contributions are summarized below:
\begin{enumerate}[a)]
    \item \emph{We show that if the system dynamics in closed-loop with the benchmark policy are \acrfull{ediss}, and the suboptimal policy is linearly converging,} then the suboptimality gap scales with the pathlength of the closed-loop suboptimal trajectory, and the convergence rate.
    \item  \emph{We derive sufficient conditions under which \acrfull*{es} of a non-smooth nonlinear time-varying system implies \acrshort{ediss}, making the condition on the benchmark policy easier to verify.}
    \item \emph{We study the suboptimal, PGM-based linear-quadratic \acrshort*{mpc} problem as an example satisfying these assumptions, and bound its suboptimality gap}.
\end{enumerate}

The \acrshort{ediss} assumption on the benchmark policy is crucial for the validity of the results. Incremental input-to-state-stability is studied in \cite{angeli2002lyapunov}, in the continuous-time and in \cite{tran2018convergence}, in the discrete-time settings, providing a condition on the deviation of two separate trajectories of the same system. If the dynamics are smooth, a sufficient condition  for \acrshort*{ediss} to hold is that of contraction \cite{lohmiller1998contraction,tsukamoto2021contraction,davydov2022non, FB-CTDS}. However, this is often not the case when one considers closed-loop dynamics under optimal controllers, e.g. under a constrained MPC \cite{bemporad2002explicit} policy. Since we are often interested in such benchmarks, we devote Section \ref{sec:sufficient_conditions}  to the derivation of sufficient conditions for \acrshort{ediss} to hold when the dynamics are non-smooth in general.

Such an incremental stability analysis allows for the derivation of asymptotically tight bounds in the sense that they scale with the \emph{level of suboptimality}, or the rate of convergence of the suboptimal policy to the optimal, converging to zero when the algorithm matches with the benchmark. The bounds also scale with the pathlength of the suboptimal trajectory, allowing an on-the-go calculation of the suboptimality gap that is independent of the benchmark states. Hence, for the suboptimal MPC example we achieve asymptotically tighter bounds under the same assumptions compared to \cite{karapetyan2023finite}. Moreover, our result is independent of the asymptotic properties of the suboptimal algorithm, providing finite-time performance bounds even when the closed-loop is not exponentially stable. 

Several notable examples of settings where such finite-time suboptimality analysis can be of use include  adaptive control \cite{hovakimyan2010L1,boffi2021regret,karapetyan2024regret}, with suboptimality due to unknown system parameters, online feedback optimization \cite{belgioioso2022online,he2022model} and online control \cite{li2019online,hazan2020nonstochastic}, with suboptimality due to unknown future costs, or real-time optimization-based control, such as real-time MPC \cite{diehl2005nominal,zanelli2020lyapunov, liao2020time}, suboptimal due to finite computational resources. 

While our analysis and results hold for any algorithm satisfying the outlined assumptions, optimization-based methods, including MPC, and their suboptimal variants are of particular interest given their wide applicability in practice. Suboptimal MPC has been studied extensively in the literature \cite{zeilinger2009real, richter2011computational,zeilinger2011real, diehl2005nominal,liao2020time,zanelli2020lyapunov,grune2010analysis,ferramosca2013cooperative}. Real-time MPC with iterative optimization methods has been studied in \cite{zeilinger2009real, richter2011computational,zeilinger2011real, diehl2005nominal,liao2020time,zanelli2020lyapunov}. Efficient warm-start methods are explored in \cite{zeilinger2011real, richter2011computational} with \cite{richter2011computational} also showing a lower bound on the number of gradient descent steps needed to achieve an open-loop suboptimality level, but no closed-loop analysis is performed. The closed-loop stability is analyzed in \cite{zeilinger2009real} using a robust MPC formulation, and also in later works in \cite{liao2020time,zanelli2020lyapunov} that extend the real-time MPC work of \cite{diehl2005nominal} and show asymptotic stability using Lyapunov methods and the small gain theorem of interconnected systems. Other sources of suboptimality for MPC have been studied in \cite{grune2010analysis}, analyzing the effect of the prediction horizon on the suboptimality of the MPC closed-loop cost, or in \cite{ferramosca2013cooperative} in the context of distributed control, to name a few. However, in this line of work the suboptimality gap as we define here is not considered. This is a common performance metric in online learning, where it is referred to as \emph{regret}. While there are several works \cite{li2019online,wabersich2022cautious, lin2022bounded} studying the regret of various model predictive controllers, their settings are different from ours. For instance, \cite{wabersich2022cautious} considers Bayesian MPC, \cite{li2019online}  studies the effect of predictions for linear time-invariant systems, and \cite{lin2022bounded} studies optimal controllers with inexact predictions using perturbation analysis, thus leading to settings that are not considered under the assumptions in this paper. The earlier work \cite{karapetyan2023finite} also analyzes the suboptimality of PGM-based MPC but does not utilize  nonlinear incremental stability analysis achieving a weaker bound that is not asymptotically tight.

The article is structured as follows. In Section \ref*{sec:preliminaries} we provide the preliminaries and the problem setup. In Section \ref*{sec:suboptimality}, we conduct the suboptimality gap analysis. Sufficient conditions for \acrshort{ediss} are derived in Section \ref*{sec:sufficient_conditions}, and in Section \ref*{sec:use_cases},  the suboptimal PGM-based linear MPC use case is studied with a numerical example.

\textit{Notation}: The sets of positive real numbers, positive integers, and non-negative integers are denoted by $\mathbb{R}_{+}$, $\mathbb{N}_+$ and $\mathbb{N}$, respectively. For some $k_0 \in \mathbb{N}$, the set of integers greater than or equal to $k_0$ is denoted by $\mathbb{N}_{\geq k_0}$. For a given vector $x$, its  Euclidean norm is denoted by $\|x\|$, and the two-norm weighted by a $Q\succ 0$  by $\|x\|_Q = \sqrt{x^{\top}Qx}$.  For a square matrix $W$, the spectral radius and the spectral norm are denoted by $\rho(W)$, and  $\|W\|$, respectively. Given a symmetric, positive definite matrix $M\succ 0$, we define its unique square-root by $M^\frac{1}{2}$, such that $M= M^\frac{1}{2} M^\frac{1}{2}$ and $M^\frac{1}{2}\succ 0$. For a matrix $W$ and $M\succ 0$, $\lambda_M^-(W)$ and $\lambda_M^+(W)$ denote the minimum and maximum eigenvalues of ${M}^{-\frac{1}{2}}W{M}^{-\frac{1}{2}}$; for any vector $x$, they satisfy $\lambda_M^{-}(W)\|x\|_M^2 \leq\|x\|_W^2 \leq \lambda_M^{+}(W)\|x\|_M^2$. The Euclidean point-to-set distance of a vector $x$ from a nonempty, closed, convex set $\mathcal{A}$ is denoted by $|x|_{\mathcal{A}} := \min_{y\in \mathcal{A}}\|x-y\|$, and the projection onto it by $\Pi_\mathcal{A}[x] = \argmin_{y\in \mathcal{A}}\|x-y\|$. The  $n-$dimensional closed ball of radius $r>0$, centered at the origin is denoted by $\mathcal{B}(r):= \{x\in \mathbb{R}^n\ | \ \|x\|\leq r\}$.  

%% file: S2_preliminaries.tex
\section{Preliminaries and Problem Setup} \label{sec:preliminaries}

We consider the optimal control problem for discrete-time, nonlinear time-varying systems of the form
\begin{equation}
    \label{eq:lti_system}
    x_{k+1} = f_0(k,x_k) + g(k,x_k,u_k),\quad k \in \mathbb{N}_{\geq k_0},
\end{equation}
where $x_k\in \mathbb{R}^n$ and $u_k \in \mathbb{R}^m$ denote the state and control input at time $k$, respectively, $f_0:\mathbb{N}_{\geq k_0} \times \mathbb{R}^n \rightarrow \mathbb{R}^n$ denotes the unforced nominal dynamics and $g: \mathbb{N}_{\geq k_0} \times \mathbb{R}^n \times \mathbb{R}^m \rightarrow \mathbb{R}^n$ the controlled dynamics. The system evolution starts from some initial state $x_{k_0} \in \mathbb{R}^n$ at some initial time $k_0\in \mathbb{N}$. The optimal control objective is to find the sequence of control inputs $\boldsymbol{u} = [u_{k_0}^{\top} \hdots u_{T-1}^{\top}]^{\top}$ that minimizes the finite-time cost
\begin{equation}
    J_T(x_{k_0},\boldsymbol{u}) = F\left(T,x_T\right) + \sum_{k=k_0}^{T-1}c(k,x_k,u_k),
     \label{eq:closed_loop_cost}
\end{equation}
where $c:\mathbb{N}_{\geq k_0}\times \mathbb{R}^n \times \mathbb{R}^m \longrightarrow \mathbb{R}$ is the stage cost at time $k$, $F:\mathbb{N}_{\geq k_0}\times \mathbb{R}^n \longrightarrow \mathbb{R}$ is the terminal cost, and $T\in \mathbb{N}_{\geq {k_0+1}}$ is the control horizon. In addition, the control input has to satisfy $u_k\in\mathcal{U}$ for all $k$, for some bounded $\mathcal{U} \subset \mathbb{R}^m$.  

An admissable policy $\pi (k,x): \mathbb{N}_{\geq k_0} \times \mathbb{R}^n \rightarrow \mathcal{U}$ maps the state at time $k$ to a  control input, generating the control signal $\boldsymbol{u}^\pi = [u_{k_0}^{\pi\top} \hdots u_{T-1}^{\pi\top}]^{\top}$  and the associated trajectory $\boldsymbol{x}^\pi= [x_{k_0}^{\pi \top} \hdots x_{T}^{\pi \top}]^{\top}$. With a slight abuse of notation, its associated cost is denoted by $J_T(x_{k_0},\pi)$. We consider time-varying systems for generality, and note that the analysis holds directly for time-invariant systems and policies as a special case. We showcase this in Section \ref*{sec:use_cases}.

We are interested in the relation of a policy $\mu$, corresponding to a given suboptimal algorithm, with respect to another benchmark policy $\mu^*$ that is equipped with desirable characteristics, e.g. optimality. The two policies are defined as follows.

\textbf{Benchmark dynamics}: Consider a benchmark policy $\mu^\star:\mathbb{N}_{\geq k_0}  \times \mathbb{R}^n \rightarrow \mathcal{U}$. Given an initial state $x_{k_0}^\star \in \mathbb{R}^n$, the benchmark dynamics are given by\footnote{For readability, we place the time $k$ in the subscripts of $\mu$, $c$, and $F$.}
\begin{equation} \label{eq:optimal_system}
    x^{\star}_{k+1} = f_0(k,x_k^{\star}) +g(k,x_k^\star,\mu_k^{\star}(x_k^{\star})):= f(k,x^{\star}_k),
\end{equation}
for all $k \in \mathbb{N}_{\geq k_0}$. We assume that there exists a set $\mathcal{D}^\star \subseteq \mathbb{R}^n$ that is forward invariant under the closed-loop dynamics \eqref{eq:optimal_system}, and restrict attention to $x_{k_0}^\star \in \mathcal{D}^\star$. Hence, $x_k^\star \in \mathcal{D}^\star$ for all $k \in \mathbb{N}_{\geq k_0}$.

\textbf{Suboptimal dynamics}: The suboptimal state evolution for a given policy $\mu:\mathbb{N}_{\geq k_0}  \times \mathbb{R}^n \rightarrow \mathcal{U}$ can be represented in the following form\footnote{We drop the explicit reference to $\mu$ from the superscript of $x$ for readability.} for any $x_{k_0} \in \mathbb{R}^n$
\begin{equation}\label{eq:suboptimal_perturbed}
    \begin{split}
        x_{k+1} = f(k,x_k) +  \underbrace{g(k, x_k, \mu_k(x_k)) - g(k,x_k, \mu_k^\star(x_k))}_{:=w_k(x(k))},
    \end{split}
\end{equation}
for all $k\in \mathbb{N}_{\geq k_0}$. The mapping $w:\mathbb{N}_{\geq k_0} \times \mathbb{R}^n\rightarrow \mathbb{R}^n$ can be thought of as a state-dependent disturbance acting on the benchmark state dynamics \eqref{eq:optimal_system}, introduced due to suboptimality. It is assumed to be such that the closed-loop suboptimal dynamics \eqref{eq:suboptimal_perturbed} evolve within a set $\mathcal{\mathcal{D}^\mu} \subseteq \mathcal{D}^\star$ that is forward invariant under \eqref{eq:suboptimal_perturbed}. Restricting  attention to initial states $x_{k_0} \in \mathcal{D}^\mu$, it then holds that  $x_k \in \mathcal{D}^\mu$ for all $k \in \mathbb{N}_{\geq k_0}$.

Figure \ref{fig:trajectories} shows the pictorial evolution of the two considered trajectories starting from the same initial state $x_0$ at time $k_0=0$.  For each $x^\star_k$, $u^{\star}_k$ denotes the control input generated by $\mu_k^\star(x_k^\star)$ and for each $x_k$, $u_k^\mu := \mu_k(x_k)$ the input generated by the suboptimal policy. 

To quantify the relation between $\mu$ and $\mu^\star$, we define the suboptimality gap of the policy $\mu$ as the additional incurred cost compared to the benchmark
\begin{equation}
    \mathcal{R}_T^{\mu}(x_{k_0}):= J_T(x_{k_0},{\mu}) - J_T(x_{k_0},\mu^{\star}),
    \label{eq:suboptimality_gap}
\end{equation}
given some $x_{k_0} \in \mathcal{D}^\mu$. An informative finite-time bound on it, which depends on fundamental quantities of the online control problem, can provide a quantifiable tradeoff between the effort needed to compute the suboptimal policy $\mu$ and the additional cost incurred by using it instead of $\mu^\star$.

We assume the benchmark policy $\mu^{\star}$ to have good performance, since otherwise,  $\mathcal{R}_T^{\mu}$ can be uninformative. We characterize this performance in terms of  \acrshort*{ediss}. 
\begin{definition}  
    A nonlinear dynamical system $x_{k+1} = f(k,x_k)$ is \acrshort*{ediss} in some forward invariant $\mathcal{D} \subseteq \mathbb{R}^n$, if there exist constants $c_0,c_w, r_w \in \mathbb{R}_+$ and $\rho \in (0,1)$, such that for any $(x_{k_0}, y_{k_0}) \in \mathcal{D} \times \mathcal{D}$, and $w_k \in \mathcal{B}(r_w), k \in \mathbb{N}_{\geq k_0}$, the perturbed dynamics $y_{k+1} = f(k,y_k)+w_k$ satisfy
    \begin{equation*}
        \|x_k-y_k\| \leq c_0 \rho^{{k-k_0}}\|x_{k_0}-y_{k_0}\|+c_w\sum_{i=k_0}^{k-1}\rho^{{k-i-1}}\|w_i\|,
    \end{equation*}
    for all $k\in \mathbb{N}_{\geq k_0}$, where the disturbances $w_k$ are such that $y_k \in \mathcal{D}$ for all $k\in \mathbb{N}_{\geq k_0}$. If $\mathcal{D} = \mathbb{R}^n$ the system is called globally \acrshort*{ediss}.
\end{definition}
\acrshort*{ediss} for continuous-time systems has been introduced in \cite{angeli2002lyapunov}. For an in-depth discussion and analysis of incremental stability in discrete-time, and its relation to contraction \cite{FB-CTDS} and convergent dynamics \cite{demidovich1967lectures}, we refer the interested reader to \cite{tran2018convergence}, and the references therein.  We assume the following for the benchmark policy.
\begin{assumption}
    (Benchmark Policy). Given the closed-loop system \eqref{eq:optimal_system}, the benchmark policy $\mu^*$ is such that
    \begin{enumerate}[label=\roman*.]
    \item \label{assum:Lipschitz} it is uniformly $L-$Lipschitz continuous in $\mathcal{D}^\star$, i.e. there exists a constant $L\in \mathbb{R}_+$, such that for all $(x,y) \in \mathcal{D}^\star \times \mathcal{D}^\star$, and $k\in \mathbb{N}_{\geq k_0}$ 
    \begin{equation*}
     \|\mu^\star_k(x)-\mu^\star_k(y)\| \leq L\|x-y\|,
    \end{equation*}
    \item  \label{assum:slow} there exist $a_k \geq 0, k\in \mathbb{N}_{\geq k_0}$, such that for all $x \in \mathbb{\mathcal{D}^\star}$ and $k\in \mathbb{N}_{\geq k_0}$
    \begin{equation*}
        \|\mu_{k+1}^\star(x) - \mu_{k}^\star(x)\| \leq a_k,
    \end{equation*} 
    \item \label{assum:ediss} the closed-loop dynamics \eqref{eq:optimal_system} are \acrshort{ediss} in $\mathcal{D}^\star$ with a rate $\rho \in (0,1)$.
\end{enumerate}
 \label{assum:mu_star}
\end{assumption}

 Although the Lipschitz continuity condition excludes benchmark policies with abrupt changes, such as discontinuities or jumps in the policy, it still holds for a relevant class of policies. In particular, when $\mu^\star$ is defined as a solution to an optimization problem, as is, for example, the MPC policy, then Lipschitz continuity holds\cite{dontchev2013euler} if the optimal solution mapping is strongly regular \cite{robinson1980strongly}. This can be verified for a range of problems by checking for strong second order sufficient conditions (SOSC), which is well studied in the literature, e.g. \cite{dontchev2019lipschitz, liao2020time}. Additionally, when state constraints are present, these can be relaxed by introducing soft constraints \cite{zeilinger2014soft,chatzikiriakos2024learning}, making the SOSC conditions easier to verify. Furthermore, we stress that the benchmark policy is not necessarily the optimal policy but any comparator policy as long as the conditions in Assumption \ref{assum:mu_star} are satisfied. The second assumption limits how fast the benchmark policy changes given the same state between two timesteps. Since  $\mathcal{U}$ is bounded, such an $a_k$ always exists, for all $k$, and can be set equal to the diameter of $\mathcal{U}$. However, it can also encode additional information, such as stationarity of the benchmark policy, in which case $a_k=0$ for all $k$. The \acrshort{ediss} condition, on the other hand, is often harder to verify, unless one assumes the existence of an incremental Lyapunov equation \cite{angeli2002lyapunov,tran2018convergence}, or uses contraction to infer incremental stability \cite{tsukamoto2021contraction}, which in turn assumes a smooth policy. Since a large class of benchmark policies, including MPC, are non-smooth globally, we show in Section \ref*{sec:sufficient_conditions}, that under further mild conditions on $f$, ES is enough to guarantee \acrshort{ediss}.

Next, we impose a linear convergence condition on the suboptimal policy $\mu$. 
\begin{assumption}
    \label{assum:mu_contractive}
    (Suboptimal Policy). Let $u_k^\mu= \mu_k(x_k)$ denote the suboptimal input evaluated on the suboptimal trajectory \eqref{eq:suboptimal_perturbed}. There exist $\eta_k \in [0,1), k\in \mathbb{N}_{\geq k_0}$ and some $u^\mu_{k_0-1}=\nu\in \mathcal{U}$, such that for all $x_{k_0} \in \mathcal{D}^\mu$, $k\in \mathbb{N}_{\geq k_0}$
    \begin{equation}\label{eq:linear_convergence}
        \|u^\mu_k-\mu^{\star}_{k}(x_k)\| \leq \eta_{k}\|u^\mu_{k-1}-\mu^{\star}_{k}(x_k)\|.
    \end{equation}
\end{assumption}
The assumption imposes at least a linear rate of convergence on the suboptimal input with respect to the benchmark given the suboptimal state. In some cases, the rate $\eta_k$ can be thought of as a design parameter that can be tuned to control the desired level of suboptimality depending on the available computational budget. Notably, such rates appear in iterative optimization algorithms, such as PGM \cite{taylor2018exact} or alternating direction method of multipliers (ADMM) \cite{nishihara2015general}, where the benchmark input is optimal with respect to some objective function, and the suboptimal input is obtained by the corresponding optimization method. When such methods are applied to optimization-based control, warm-started with the previous input, the exact linear convergence form of \eqref{eq:linear_convergence} appears in \cite{zanelli2020lyapunov,liao2020time,liao2021analysis} for PGM and in \cite{srikanthan2024closed} for ADMM. In Section \ref*{sec:use_cases}, we consider the PGM-based MPC setting of \cite{liao2021analysis} and show how both assumptions are verified.

We restrict our attention to systems where the controlled dynamics $g$ are Lipschitz continuous with respect to $u$, uniformly in $x$ and $k$. 
\begin{assumption}\label{assum:Lipschitz_g}
    There exists a constant $L_u \in \mathbb{R}_+$, such that for any $(u,v) \in \mathbb{R}^m \times \mathbb{R}^m$, for all $x\in \mathcal{D}^\mu$ and $k\in \mathbb{N}_{\geq k_0}$
    \begin{equation*}
        \|g(k,x,u) - g(k,x,v)\| \leq L_u \|u - v\|.
    \end{equation*}
\end{assumption}
This is satisfied, for instance, in linear time-invariant systems or control-affine nonlinear systems of the form $x_{k+1} = f(x_k) + g(x_k)u_k$, often studied in the context of feedback linearization (see \cite{khalil2002nonlinear,haddad2008nonlinear} for details). Finally, we restrict our analysis to local Lipschitz continuous stage costs. 
\begin{assumption}
    \label{assum:stage_cost}
    There exist constants $M_x,M_u \in \mathbb{R}_+$, such that for all $(x,y) \in \mathcal{D}^\mu \times \mathcal{D}^\mu$, $(u,z) \in \mathcal{U} \times \mathcal{U}$ and $k\in \mathbb{N}_{\geq k_0}$
    \begin{align*}
        \|c_k(x,u) - c_k(y,z)\| &\leq M_x \|x-y\|+M_u\|u-z\|,\\
        \|F_k(x) - F_k(y)\|\ &\leq \ M_x\|x-y\|.
    \end{align*}
\end{assumption}

%% file: S3_suboptimality_gap.tex
\section{Suboptimality Gap Analysis} \label{sec:suboptimality}
In this section, we analyze the suboptimality gap for a given policy and show that $\mathcal{R}_T$ scales with the product of the pathlength of the suboptimal dynamics and a vector dependent on the convergence rates.  In the analysis we take $k_0=0$ without loss of generality. We define the backward difference path vector, ${\Delta}\in \mathbb{R}^{T-1}$, to be 
\begin{equation*}
    {\Delta}: = 
    \begin{bmatrix}
        \|\delta x_1\| & \|\delta x_2\| & \hdots & \|\delta x_{T-1}\|
    \end{bmatrix}^\top,
\end{equation*} 
where $\delta x_k = x_k-x_{k-1},\; k\in \mathbb{N}_+$, where $x_k$ is the state at time $k$ for the suboptimal dynamics \eqref{eq:suboptimal_perturbed}. The pathlength of the suboptimal trajectory is then defined as $\mathcal{S}_T = \|{\Delta}\|_1$ and the Euclidean pathlength as $\mathcal{S}_{T,2}:= \|{\Delta}\|$.

The policy convergence rate vector, ${\tilde{\eta}}\in \mathbb{R}^{T-1}$ is defined as
\begin{equation*}
     {\tilde{\eta}}: = 
     \begin{bmatrix}
        \tilde{\eta}_1 & \tilde{\eta}_2 & \hdots & \tilde{\eta}_{T-1}
    \end{bmatrix}^\top,
\end{equation*} 
where 
\begin{equation*}
    \tilde{\eta}_k := \sum_{i=k}^{T-1}\prod_{j=k}^i\eta_j, \qquad \forall k \in [0,T-1].
\end{equation*}
Note that $\tilde{\eta}_k = \mathcal{O}(\eta_k)$ and provides a weighting on the influence of the $\delta x_k$ on $\mathcal{R}_T$. This is analyzed further in Section \ref*{sec:analysis}. We denote the Euclidean norm of the suboptimality vector by $\bar{\eta}:= \|{\tilde{\eta}}\|$. The rate of change of the benchmark input $\mu^\star(x^\star)$ is captured by the vector $a \in \mathbb{R}_+^{T-1}$, defined as
\begin{equation*}
    a: =\begin{bmatrix}
        a_1 & a_2 & \hdots & a_{T-1}
    \end{bmatrix}^\top.
\end{equation*}

\subsection{Upper Bound}
The bound in the following theorem captures the tradeoff between suboptimality and  the additional  cost in closed-loop.
\begin{theorem}
\label{the:incurred_suboptimality}
    Let Assumptions \ref{assum:mu_star}, \ref*{assum:Lipschitz_g} and \ref*{assum:stage_cost} hold, then the suboptimality gap  of any policy, $\mu$, fulfilling Assumption \ref{assum:mu_contractive} satisfies
    \begin{equation*}
        \mathcal{R}_T^{\mu}(x_0) =\mathcal{O}\left(\left(a + {\Delta} \right)^\top \tilde{\eta}\right),
     \end{equation*}
    for all $x_0 \in \mathcal{D}^\mu$. Specifically, it is bounded by
    \begin{equation*}
        \mathcal{R}_T^\mu(x_0) \leq \bar{M}\left(\tilde{\eta}_0\|\delta u_0\|+ \left(a +  L{\Delta} \right)^\top{\tilde{\eta}} \right),
    \end{equation*}
    for all $x_0 \in \mathcal{D}^\mu$, where $\delta u_0 := \nu - \mu_0^\star(x_0)$ and $ \bar{M} := \left(M_u+\frac{c_wL_u\left(M_uL+M_x\right)}{1-\rho}\right)$.
\end{theorem}
The asymptotic variable in the $\mathcal{O}(\cdot)$ notation is the time horizon length $T$, with all terms in the product $\left(a + {\Delta} \right)^\top \tilde{\eta}$ scaling  with $T$, as  $T\rightarrow \infty$. The bound tends to zero as $\tilde{\eta}$ decreases. This is intuitive, as smaller $\tilde{\eta}$ suggests that the benchmark and suboptimal  trajectories are closer to each other. Additionally, the suboptimality gap is a relative measure, but the bound is fully decoupled from the performance of $\mu^{\star}$ and only depends on the performance of the suboptimal state evolution if $a=\boldsymbol{0}$. In this case, the bound can also be given as $\mathcal{R}_T(x_0) = \mathcal{O}(\bar{\eta}\mathcal{S}_{T,2})$, where $\mathcal{S}_{T,2}$ captures the transient behavior of the suboptimal system and is well-defined in the limit as ${T\!\rightarrow\! \infty}$ when, for example, \eqref{eq:suboptimal_perturbed} is exponentially stable. 

Before we prove Theorem \ref*{the:incurred_suboptimality}, we introduce several supporting lemmas, and the Cauchy Product inequality defined for two finite series $\{a_i\}_{i=1}^T$ and $\{b_i\}_{i=1}^T$ as follows
\begin{equation}
    \label{eq:cauchy_product}\textstyle{
    \sum_{i=0}^T\left|\sum_{j=0}^{i}a_jb_{i-j}\right| \leq  \left(\sum_{i=0}^T|a_i|\right) \left(\sum_{j=0}^T|b_j|\right)}.
\end{equation}
We let $u_k^\mu = \mu_k(x_k)$ denote the suboptimal input as in Assumption \ref{assum:mu_contractive}. 

\begin{lemma} \label{lem:delta_mu}
    Let Assumption \ref{assum:mu_star} hold, then for any policy, $\mu$, fulfilling Assumption \ref{assum:mu_contractive}, and for all $x_0 \in \mathcal{D}^\mu$
    \begin{equation}\label{eq:d_k_bound}
        \begin{split}
        \sum_{k=0}^{T-1}\|d_k\| \leq \tilde{\eta}_0\|\delta u_0\|+ \left(a +  L{\Delta} \right)^\top {\tilde{\eta}},
        \end{split}
    \end{equation}
    where we define $d_k := u_k^\mu- \mu_k^{\star}(x_k),$ and $\delta u_0:= \nu-\mu^{\star}_0(x_0)$.
\end{lemma}

\begin{proof}
    Consider the bound on the per-step error in the input
    \begin{align*}
        \|d_k\| &\stackrel{{(a)}}{\leq} \eta_k\|u_{k-1}^\mu-\mu^{\star}_{k}(x_k)\|\\
        \begin{split}
        &\stackrel{{(b)}}{\leq} \eta_k\|u_{k-1}^\mu-\mu^{\star}_{k-1}(x_{k-1})\|\\
         &\qquad \qquad \qquad+ \eta_k\|\mu^{\star}_{k-1}(x_{k-1})- \mu^{\star}_{k}(x_k)\| 
        \end{split}\\
        \begin{split}
            &\stackrel{{(c)}}{\leq} \eta_k\|d_{k-1}\|+ \eta_k\|\mu^{\star}_{k-1}(x_{k-1}) - \mu^{\star}_{k-1}(x_{k})\|\\
             &\qquad \qquad \qquad + \eta_k\|\mu^\star_{k-1}(x_k) - \mu^\star_k(x_k)\|
            \end{split}\\
        \begin{split}
        &\stackrel{{(d)}}{\leq} \eta_k\|d_{k-1}\|+ \eta_k  L \|x_k-x_{k-1}\| + \eta_k  a_k,
        \end{split}
    \end{align*}
the inequality $(a)$ follows directly from  Assumption \ref*{assum:mu_contractive}, $(b)$ and $(c)$ follow from the triangle inequality for vector norms and $(d)$ from the uniform Lipschitz condition in Assumption \ref*{assum:mu_star}.\ref*{assum:Lipschitz} and Assumption \ref*{assum:mu_star}.\ref*{assum:slow}. Applying the above inequality recursively leads to
\begin{align*}
    \|d_k\|&\leq \|d_0\|\prod_{i=1}^k \eta_i + \sum_{j=1}^{k}\left(a_j+L\|\delta x_{j}\|\right)\prod_{i=j}^k \eta_i,
\end{align*}
for all $k\in \mathbb{N}_+$. Summing up for $k=0, \hdots T-1$
\begin{equation*}
    \begin{split}
     \sum_{k=0}^{T-1}\|d_k\|&\leq \|\delta u_0\| \sum_{k=0}^{T-1}\prod_{i=0}^k \eta_i +   \sum_{k=1}^{T-1}\sum_{j=1}^{k}\left(a_j+L\|\delta x_{j}\|\right)\prod_{i=j}^k \eta_i\\
     &= \tilde{\eta}_0\|\delta u_0\|+ \left(a +  L{\Delta} \right)^\top{\tilde{\eta}} ,
    \end{split}
 \end{equation*}
where the inequality follows from Assumption \ref*{assum:mu_contractive} and the equality from the definitions of $\tilde{\eta}$, $a$ and $\Delta$.
\end{proof}

The following lemma provides an upper bound on the norm of the finite-time trajectory mismatch.
\begin{lemma}
    \label{lem:delta_x}
     Let Assumptions \ref{assum:mu_star} and \ref*{assum:Lipschitz_g} hold, and consider any policy $\mu$ fulfilling Assumption \ref{assum:mu_contractive}. Then, there exist constants $c_0,c_W \in \mathbb{R}_+$, such that for all $(x_0, x_0^\star) \in (\mathcal{D}^\mu \times \mathcal{D}^\mu)$
     \begin{align*}
        &\sum_{k=0}^{T}\|x_k-x^{\star}_k\| \leq \|x_0 - x_0^{\star}\|\left(\frac{c_0\left(1-\rho^{T+1}\right)}{1-\rho} \right)\\
        &\qquad+ \frac{c_w L_u\left(1-\rho^T\right)}{1-\rho}\left(\tilde{\eta}_0\|\delta u_0\|+ \left(a +  L{\Delta} \right)^\top {\tilde{\eta}}\right),
     \end{align*}
     where $x_k$ and $x_k^{\star}$ denote the states at time $k$ under, respectively, the suboptimal and benchmark policies.
\end{lemma}
\begin{proof}
    Given the boundedness of $\mathcal{U}$, the uniform Lipschitz continuity of $g$ in $u$, and recalling the definition of $w_k$ from \eqref{eq:suboptimal_perturbed}, it follows that there exists a $r_w \in \mathbb{R}_+$, such that $w_k(x_k) \in \mathcal{B}(r_w), k \in \mathbb{N}$. Consider then the dynamics \eqref{eq:suboptimal_perturbed} as a perturbed version of the benchmark dynamics \eqref{eq:optimal_system}. Assumption \ref{assum:mu_star} ensures that there exist $c_0,c_w \in \mathbb{R}_+$ and $\rho \in (0,1)$, such that for all $k\in \mathbb{N}$ and $(x_0,x_0^{\star}) \in \mathcal{D}^\mu \times \mathcal{D}^\mu$
    \begin{align*}
        \|x_k-x_k^{\star}\| &\leq c_0 \rho^k\|x_0-x_0^{\star}\| +c_w \sum_{i=0}^{k-1}\rho^{k-i-1}\|w_i(x_i)\|\\
        &\leq  c_0 \rho^k\|x_0-x_0^{\star}\| +c_w L_u\sum_{i=0}^{k-1}\rho^{k-i-1}\|d_i\|, 
    \end{align*}
    where the second inequality follows from Assumption \ref*{assum:Lipschitz_g}, where $d_k := u_{k}^\mu- \mu_k^{\star}(x_k)$. Summing up over the whole trajectory and noting the resultant finite geometric series
    \begin{align*}
        &\sum_{k=0}^{T}\|x_k-x^{\star}_k\| \\
        &\; \leq\|x_0 - x_0^{\star}\|\left(\frac{c_0\left(1-\rho^{T+1}\right)}{1-\rho} \right) + c_w L_u\sum_{k=0}^{T-1}\sum_{i=0}^{k} \rho^{k-i}\|d_i\| \\
        &\; \leq  \|x_0 - x_0^{\star}\|\left(\frac{c_0\left(1-\rho^{T+1}\right)}{1-\rho} \right) + c_w L_u\sum_{k=0}^{T-1}\rho^{k}\;\sum_{i=0}^{T-1} \|d_i\|,
    \end{align*}
    where the last inequality follows from \eqref{eq:cauchy_product}. Finally, the result follows by using the bound \eqref{eq:d_k_bound} in Lemma \ref*{lem:delta_mu}.
 \end{proof}

 Similarly, the norm of the finite-time input signal mismatch due to the difference in applied inputs can be bounded in the following Lemma.
 \begin{lemma}
    \label{lem:delta_u}
     Let Assumptions \ref{assum:mu_star} and \ref*{assum:Lipschitz_g} hold,  and consider any policy $\mu$ fulfilling Assumption \ref{assum:mu_contractive}. Then, there exist constants $c_0,c_w \in \mathbb{R}_+$, such that for all $(x_0, x_0^\star) \in (\mathcal{D}^\mu \times \mathcal{D}^\mu)$ 
     \begin{align*}
        &\sum_{k=0}^{T-1}\|u_{k}^\mu-\mu_k^{\star}(x_k^{\star})\| \leq \|x_0-x_0^{\star}\|\frac{Lc_0\left(1-\rho^{T}\right)}{1-\rho} \\
        &+\left(\tilde{\eta}_0\|\delta u_0\|+ \left(a +  L{\Delta} \right)^\top {\tilde{\eta}}\right) \left(1 + \frac{Lc_w L_u\left(1-\rho^{T-1}\right)}{1-\rho}\right), 
    \end{align*}
    where $x_k$ and $x_k^{\star}$ are the states at time $k$ under, respectively, the suboptimal and benchmark policies, $\mu$ and $\mu^{\star}$.
\end{lemma}
\begin{proof}
    Using the triangle inequality for vector norms and defining $d_k := u_{k}^\mu- \mu_k^{\star}(x_k)$
    \begin{align*}
        \|u_{k}^\mu-\mu_k^{\star}(x_k^{\star})\| &\leq \|d_k\|+\|\mu_k^{\star}(x_k) - \mu_k^{\star}(x_k^{\star})\| \\
        & \leq \|d_k\| + L \|x_k-x_k^{\star}\|,
    \end{align*}
    where the last inequality follows from Lipschitz continuity of $\mu^{\star}$ from Assumption \ref*{assum:mu_star}.\ref*{assum:Lipschitz}. Summing up over the trajectory horizon, using the bound \eqref{eq:d_k_bound} from Lemma \ref*{lem:delta_mu} and the bound in Lemma \ref*{lem:delta_x} completes the proof.
\end{proof}
\begin{proof}[Proof of Theorem \ref*{the:incurred_suboptimality}]
    By Assumption \ref*{assum:stage_cost}, for all $(x_0,x_0^\star) \in \mathcal{D}^\mu \times \mathcal{D}^\mu$
    \begin{align*}
        &\mathcal{R}_T(x_0,x_0^\star): = J_T(x_0,{\mu}) - J_T(x_0^\star,\mu^{\star} )\\
        &\qquad \leq M_x\sum_{k=0}^{T} \|x_k-x_k^{\star}\| + M_u \sum_{k=0}^{T-1}\|u_{k}^\mu- \mu_k^{\star}(x_k^{\star})\|. 
    \end{align*}
    Then, using Lemmas \ref*{lem:delta_x} and \ref*{lem:delta_u} for the two respective sums
    \begin{align*}
        \mathcal{R}_T(x_0,x_0^\star) &< \|x_0-x_0^{\star}\|\frac{c_0\left(M_uL+M_x\right)}{1-\rho} \\
        &\qquad +\bar{M}\left(\tilde{\eta}_0\|\delta u_0\|+ \left(a +  L{\Delta} \right)^\top {\tilde{\eta}}\right).
    \end{align*}
    The result follows by taking $x_0 = x_0^\star$.
\end{proof}

For the special case when the convergence rate of the suboptimal policy is constant, the following Corollary follows.

\begin{corollary} \label{cor:incremental_stability_corollary}
    Let Assumptions \ref{assum:mu_star}, \ref*{assum:Lipschitz_g} and \ref*{assum:stage_cost} hold, and consider a suboptimal policy, $\mu$, fulfilling Assumption \ref{assum:mu_contractive} with $\eta_k =\eta,\; k\in \mathbb{N}$. Then, its suboptimality gap is given by
    \begin{equation*}
        \mathcal{R}_T^{\mu}(x_0) = \mathcal{O}\left(\frac{\eta}{1-\eta}  \left(\mathcal{S}_T+\|a\|_1\right)\right),
    \end{equation*}
    for all $x_0 \in \mathcal{D}^\mu$.
\end{corollary}
\begin{proof}
    If $\eta_k = \eta,\; k\in \mathbb{N}$, $\tilde{\eta}_k = \eta \left(1-\eta^{T-k}\right)/\left(1-\eta\right),$ and is  bounded by
    \begin{equation*}
        \tilde{\eta}_k\leq \frac{\eta}{1-\eta},\quad k \in [0,T].
    \end{equation*}
The complexity term then satisfies
\begin{equation*}
    \left(a+{\Delta}\right)^\top {\tilde{\eta}} \leq \frac{\eta \left(\mathcal{S}_T +\|a\|_1\right) }{1-\eta}.
\end{equation*}
The rest of the proof follows directly from Theorem \ref*{the:incurred_suboptimality} by replacing the complexity term with the new bound.
\end{proof}

For bounded $\mathcal{S}_T$, the corollary shows the effect of $\eta$ on the suboptimality gap. Crucially, the gap is $0$ in the absence of suboptimality. On the other hand, given a fixed $\eta$, one can calculate an upper bound on the suboptimality gap of the policy $\mu$ by only considering its pathlength $\mathcal{S}_T$. The independence from the pathlength of the benchmark policy is particularly useful for online suboptimality estimation since it cannot be generally accessed by the control designer.

\subsection{Interpretation of the Upper Bound} \label{sec:analysis}

The term $\tilde{\eta}_0\|\delta u_0\|$ in the bound of Theorem \ref*{the:incurred_suboptimality} captures the error due to the initial mismatch in the control input, $\delta u_0$. This term in general cannot be avoided, unless the initial ``guess" of the input $\nu$ is correct, or $\eta_0=0$, so that the suboptimal and benchmark policies match at the initial timestep.

The second term, $a^\top \tilde{\eta}$, scales with the magnitude of the rate of change of the time-varying benchmark policy $\mu^\star$, as defined in Assumption \ref*{assum:mu_star}.\ref*{assum:slow}. It vanishes either when $\mu^\star$ is stationary, or when the benchmark and suboptimal policies coincide.

\begin{figure}[ht]
    \centering
    \begin{subfigure}{.49\columnwidth}
      \centering\captionsetup{width=\linewidth}
      \includegraphics[width=\linewidth]{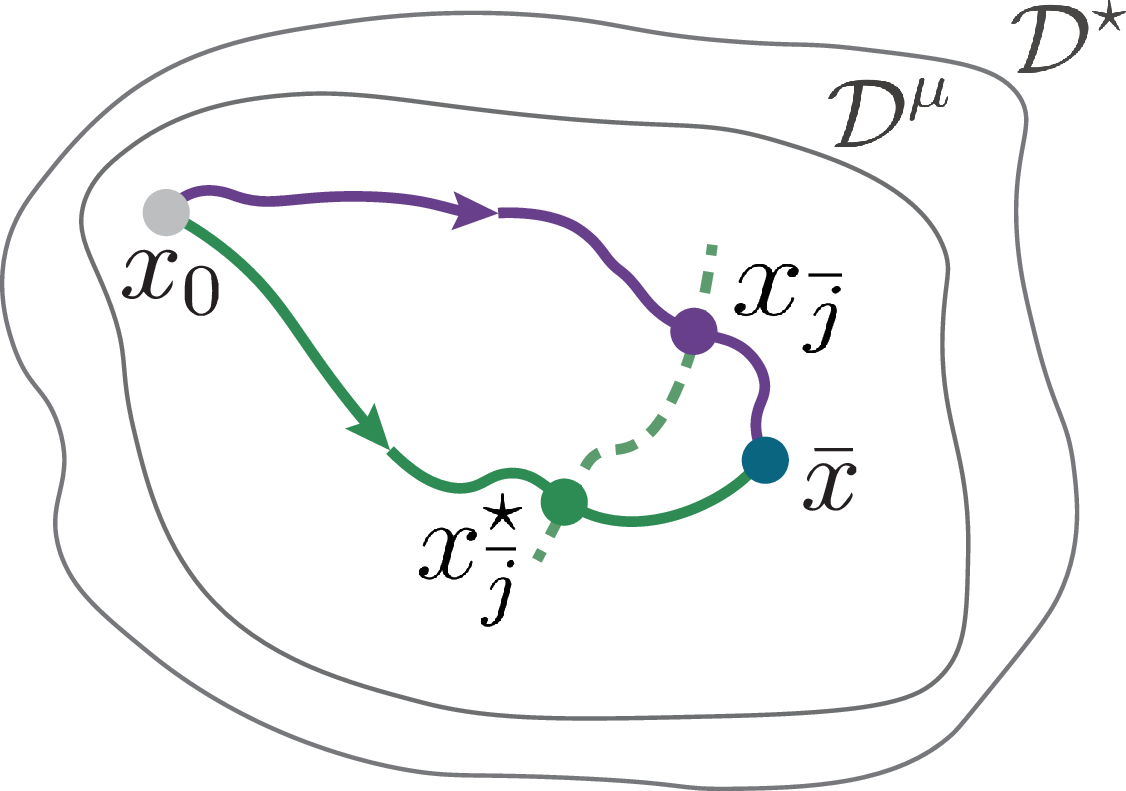}
      \caption{Exponentially stable case. Suboptimality is captured by \eqref{eq:equilibrium}.}
      \label{fig:equilibrium}
    \end{subfigure}
    \begin{subfigure}{.49\columnwidth}
        \centering\captionsetup{width=.95\linewidth}
        \includegraphics[width=\linewidth]{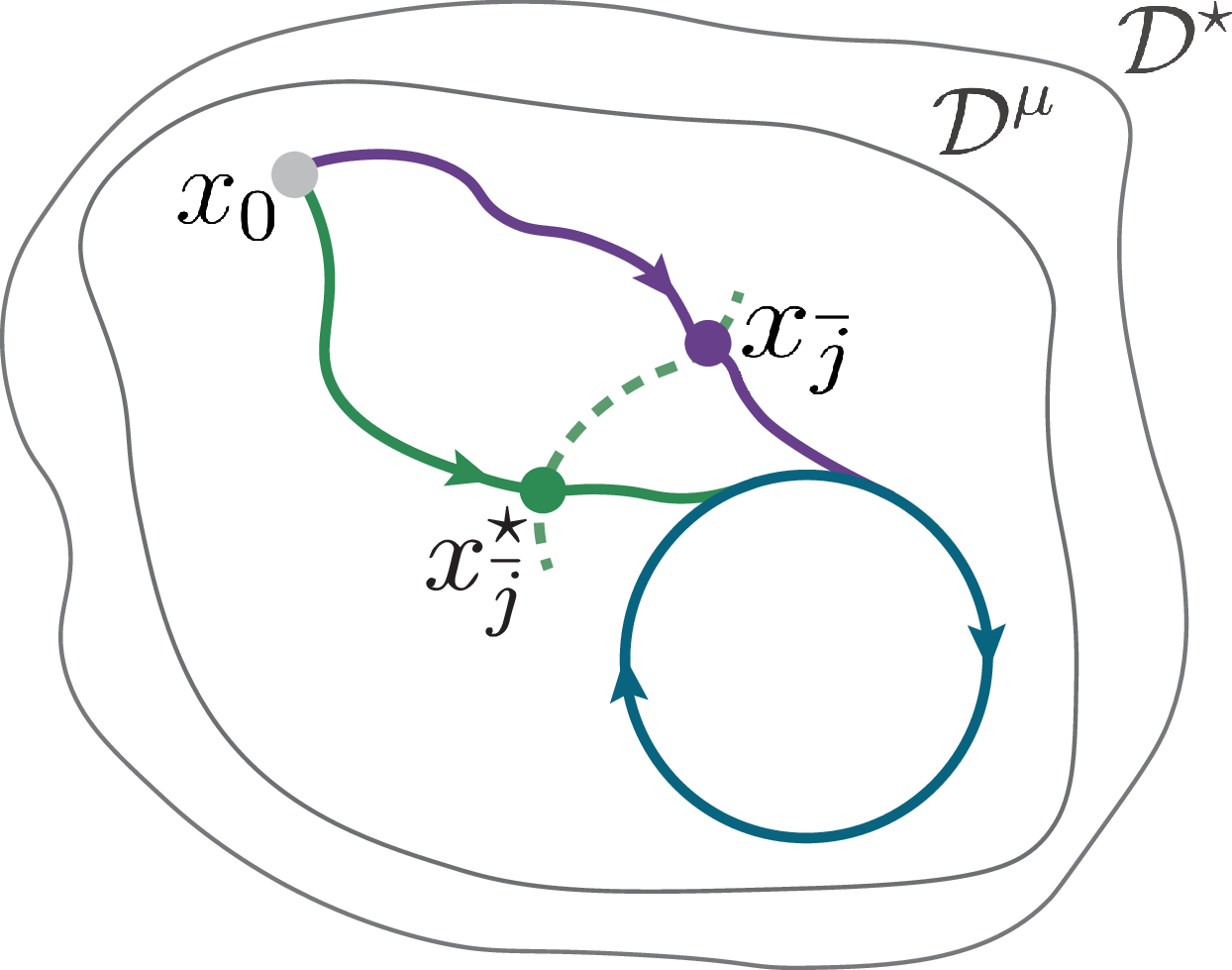}
        \caption{Case of limit cycles. Suboptimality is captured by \eqref{eq:limit_cycle}.}
        \label{fig:limit_cycle}
      \end{subfigure}%
    \caption{The pictorial evolution of suboptimal and benchmark trajectories evolving in $\mathcal{D}^\mu \subseteq \mathcal{D}^\star$.}
    \label{fig:analysis}
\end{figure}

The main complexity term of interest, $\Delta^\top \tilde{\eta}$ captures the suboptimality of the policy through the inner product of the path vector ${\Delta}$ and the suboptimality vector ${\tilde{\eta}}$. To study the interplay of these two quantities in more detail, let us consider the case when the benchmark dynamics \eqref{eq:optimal_system} have an equilibrium at some $\bar{x}\in \mathbb{R}^n$, and $k_0=0$ without loss of generality. If the suboptimal closed-loop system \eqref{eq:suboptimal_perturbed} is also exponentially stable with $\eta_k \neq 0, \; \forall k\in \mathbb{N}$, then $\|\delta x_j\|\approx 0,\; \forall j\geq \bar{j}$, for some $\bar{j} \in \mathbb{N}$, as visualised in Figure \ref*{fig:equilibrium}. Note that, in order to provide this interpretation it is assumed that the systems \eqref{eq:optimal_system} and \eqref{eq:suboptimal_perturbed} have the same single equilibrium point $\bar{x}$. In such a setting, the complexity term captures the fact that the suboptimality gap is finite as follows 
\begin{equation}
    \label{eq:equilibrium}
    \begin{split}
    &{\Delta}^\top {\tilde{\eta}}=\\
    &\left[
        \begin{array}{c c|c c c}
            \|\delta x_1\| \; \hdots \; \color{gray}\undermat{\approx \boldsymbol{0}}{\color{gray}\|\delta x_{\bar{j}}\| \; \color{gray}\hdots \; \color{gray}\|\delta x_{T-1}\|}
        \end{array}
    \right]
    \left[
        \begin{array}{c}
            \Vspace\tilde{\eta}_1\\
            \Vspace \vdots \\ 
            \hline 
            \Vspace \tilde{\eta}_{\bar{j}}\\
            \Vspace \vdots \\ 
            \Vspace \tilde{\eta}_{T-1}
        \end{array}
    \right]
    \begin{array}{@{\kern-\nulldelimiterspace}l@{}}
        \begin{array}{@{}c@{}}\Vspace\\\Vspace \end{array} \\
        \left.\begin{array}{@{}c@{}}\Vspace\\\Vspace \\ \Vspace\\\end{array}\right\}\neq 0
      \end{array}.
    \end{split}
\end{equation}
This example coincides with the suboptimal MPC use case discussed in detail in Section \ref*{sec:use_cases}. When there are multiple equilibria, a similar interpretation of the bound holds. Among other possibilities, one can also consider the case when the benchmark dynamics \eqref{eq:optimal_system} converge to a limit cycle. Since \eqref{eq:optimal_system} is \acrshort*{ediss} it follows from \eqref{eq:suboptimal_perturbed} that if at a given point in time $\bar{j}\in \mathbb{N}$, the suboptimal policy coincides with the considered benchmark, i.e. $\eta_j =0, \; \forall j\geq \bar{j}$, then the trajectories will necessarily coincide, as visualized in Figure \ref*{fig:limit_cycle}. This is captured by the complexity term as 
\begin{equation}
    \label{eq:limit_cycle}
    \begin{split}
    &{\Delta}^\top {\tilde{\eta}}=\\
    &\left[
        \begin{array}{c c|c c c}
            \|\delta x_1\| \; \hdots \; \undermat{\neq \boldsymbol{0}}{\|\delta x_{\bar{j}}\| \; \hdots \; \|\delta x_{T-1}\|}
        \end{array}
    \right]
    \left[
        \begin{array}{c}
            \Vspace\tilde{\eta}_1\\
            \Vspace \vdots \\ 
            \hline 
            \color{gray} \Vspace \tilde{\eta}_{\bar{j}}\\
            \color{gray} \Vspace \vdots \\ 
            \color{gray} \Vspace \tilde{\eta}_{T-1}
        \end{array}
    \right]
    \color{gray}
    \begin{array}{@{\kern-\nulldelimiterspace}l@{}}
        \begin{array}{@{}c@{}}\Vspace\\\Vspace \end{array} \\
        \left.\begin{array}{@{}c@{}}\Vspace\\\Vspace \\ \Vspace\\\end{array}\right\}=\boldsymbol{0}
      \end{array}\color{black}{.}
    \end{split}
\end{equation}
Even though the pathlength keeps increasing, the norm of the suboptimality vector is finite, resulting in a finite suboptimality gap, containing only the additional cost due to suboptimality at the first $\bar{j}$ timesteps.

%% file: S4_incremental_stability.tex
\section{Exponentially Stable Policies and \acrshort*{ediss}} \label{sec:sufficient_conditions}

The analysis in the previous section is agnostic to the steady-state properties of \eqref{eq:optimal_system} and, therefore, has general validity. In this section, we analyze a common case when the dynamics \eqref{eq:optimal_system} have a unique equilibrium point at the origin. We derive sufficient conditions under which the exponential stability of globally  non-smooth nonlinear dynamics implies \acrshort*{ediss}, thus verifying the otherwise non-trivial Assumption \ref{assum:mu_star}.\ref{assum:ediss}.

We treat \eqref{eq:optimal_system} as a general nonlinear time-varying system of the form
\begin{equation} \label{eq:nonlinearsystem}
    x_{k+1} = f(k,x_k), \quad  k \in \mathbb{N}_{\geq k_0},
\end{equation}
where $f:\mathbb{N}_{\geq k_0} \times \mathcal{D} \rightarrow \mathcal{D}$ is continuous with respect to both arguments, $x_{k_0} = \xi \in \mathbb{R}^n$ for some $k_0 \in \mathbb{N}$, $\mathcal{D}\subset \mathbb{R}^n$. Moreover, the set $\mathcal{D}$ is assumed to be compact and to contain the origin. The solution of the system \eqref{eq:nonlinearsystem} at time $k \in \mathbb{N}_{\geq k_0}$ is characterized by the function $\phi:\mathbb{N}_{\geq k_0} \times \mathbb{N}_{\geq k_0} \times \mathcal{D}\rightarrow\mathcal{D}$ mapping the current time,  initial time and the initial state to the current state, i.e.   $\phi(k+1, k_0, \xi) = f(k,\phi(k,k_0,\xi))$ for all  $k \in \mathbb{N}_{\geq k_0}$. We consider the origin to be an equilibrium point for \eqref{eq:nonlinearsystem}, i.e. $f(k,\boldsymbol{0}) = \boldsymbol{0}$ for all $k\in \mathbb{N}_{\geq k_0}$. Although this restricts the attention to regulation problems, one can always recast a tracking problem into a regulation one by considering the nonlinear evolution of the error between the state and the reference. We impose the following assumption.
\begin{assumption} (Local Behavior). \label{assum:local_behavior}
    The dynamics $f$ in \eqref{eq:nonlinearsystem} are
    \begin{enumerate}[label=\roman*.]
        \item \label{assum:local_lipschitz}
        $L_f$-Lipschitz continuous in  $\mathcal{D} \subset \mathbb{R}^n$,  
        \item \label{assum:local_diff} continuously differentiable with respect to $x$, and the Jacobian matrix $[\partial f(k,x)/ \partial x]$ is bounded and Lipschitz continuous uniformly in $k$ within some region $\mathcal{D}_{0}\subseteq \mathcal{D}$ containing the origin.
    \end{enumerate}
\end{assumption}

For completeness,  we present a series of auxiliary definitions and theorems for (incremental) exponential stability before presenting the main results of this section.

\subsection{Preliminaries on Exponential Stability}

We define uniform exponential stability for discrete-time, nonlinear time-varying systems \cite{haddad2008nonlinear}.

\begin{definition}\label{def:exponential_stability}
    Given the system \eqref{eq:nonlinearsystem}, the equilibrium point $x=\boldsymbol{0}$ is  uniformly exponentially stable in $\mathcal{D}$ with a rate $\lambda$,  if there exist constants $d \in \mathbb{R}_+$ and $\lambda \in (0,1)$, such that for all $\xi \in \mathcal{D}$, and $k \in \mathbb{N}_{\geq k_0}$    
    \begin{equation}\label{eq:exponential_stability}
        \|\phi(k,k_0,\xi)\| \leq d\|\xi\|\lambda^{k-k_0}.    
    \end{equation}
    If $\mathcal{D}= \mathbb{R}^n$, then the equilibrium is uniformly globally exponentially stable.
\end{definition}

If the origin is exponentially stable, we also refer to the system \eqref{eq:nonlinearsystem} as such.  Lyapunov theory provides necessary and sufficient conditions for the exponential stability of nonlinear systems. Below are the discrete-time Lyapunov theorems for exponential stability.

\begin{theorem}\cite[Thm. 13.11]{haddad2008nonlinear} \label{the:direct_lyapunov}
    If there exists a continuous mapping $V: \mathbb{N}_{\geq k_0} \times \mathcal{D} \rightarrow \mathbb{R}_+$, and some constants $c_1,c_2 \in \mathbb{R}_+$, $\beta \in (0,1)$ and $p\geq 1$, such that for all $\xi \in \mathcal{D}$, and $k \in \mathbb{N}_{\geq k_0}$
    \begin{align*}
        c_1 \|\xi\|^p \leq V(k,\xi) \leq c_2 \|\xi\|^p,\\
        V\left(k+1,f(k,\xi)\right) \leq \beta^p V(k,\xi),
    \end{align*}
    then the nonlinear system \eqref{eq:nonlinearsystem} is uniformly  exponentially stable in $\mathcal{D}$, with a rate $\beta$.
\end{theorem}
The converse Lyapunov theorem for the discrete-time case shows the implication in the opposite direction.
\begin{theorem} \label{the:converse_lyapunov}
    If the nonlinear system \eqref{eq:nonlinearsystem} is uniformly exponentially stable in $\mathcal{D}$, then there exists a continuous function $V:\mathbb{N}_{\geq k_0} \times \mathcal{D} \rightarrow \mathbb{R}$ and constants $c_1,c_2 \in \mathbb{R}_+$ and  $\beta \in (0,1)$, such that for all $\xi \in \mathcal{D}$, and $k \in \mathbb{N}_{\geq k_0}$
    \begin{align} \label{eq:converse_lyapunov_1}
        c_1 \|\xi\|^2 \leq V(k,\xi) \leq c_2 \|\xi\|^2,\\
        \label{eq:converse_lyapunov_2}
        V\left(k+1,f(k,\xi)\right) \leq \beta^2 V(k,\xi).
    \end{align}
\end{theorem}
{The proof for Theorem \ref{the:converse_lyapunov} follows directly from \cite[Thm 2]{jiang2002converse}, by adapting the global result to the case with forward invariant subspace $\mathcal{D}$. The above theorems generalize to uniform global exponential stability if $\mathcal{D} = \mathbb{R}^n$ \cite{khalil2002nonlinear,jiang2002converse}. In continuous-time, the rate of change of the Lyapunov function with respect to the state is bounded by the norm of the state  \cite[Thm 4.14]{khalil2002nonlinear}. The following Lemma is the discrete-time equivalent of this.
\begin{lemma} \label{lem:delta_V_bound}
    Let Assumption \ref*{assum:local_behavior}.\ref*{assum:local_lipschitz} hold, then if the system \eqref{eq:nonlinearsystem} is  uniformly exponentially stable in $\mathcal{D}$, there exists a constant $c_3\in \mathbb{R}_+$ and a continuous Lyapunov function $V:\mathbb{N}_{\geq k_0}\times \mathcal{D} \rightarrow \mathbb{R}$, that, in addition to \eqref{eq:converse_lyapunov_1} and \eqref{eq:converse_lyapunov_2}, for all $ (\xi,\zeta) \in \mathcal{D} \times \mathcal{D}$, and $k \in \mathbb{N}_{\geq k_0}$ also satisfies 
    \begin{equation*}
        |V(k,\xi)-V(k,\zeta)| \leq c_3 \|\xi-\zeta\|\left(\|\xi\|+\|\zeta\|\right).
    \end{equation*}
\end{lemma}
The proof of Lemma \ref*{lem:delta_V_bound} is provided in the appendix. 

\subsection{Preliminaries on Exponential Incremental Stability}

Exponential incremental stability shows the exponential convergence of two trajectories generated by the same system to each other \cite{tran2018convergence,tsukamoto2021contraction}. 

\begin{definition}\label{def:incremental_stability}
 The system \eqref{eq:nonlinearsystem} is uniformly   exponentially incrementally stable in  $\mathcal{D}$ if there exist constants $d\in \mathbb{R}_+$ and $\lambda\in(0,1)$, such that for all  $(\xi,\zeta) \in \mathcal{D} \times \mathcal{D}$ and $k \in \mathbb{N}_{\geq k_0}$
 \begin{equation*}
    \|\phi(k,k_0,\xi)-\phi(k,k_0,\zeta)\| \leq d\|\xi-\zeta\|\lambda^{k-k_0}.
 \end{equation*}
 If this holds only for a subset of initial conditions in $\mathcal{D}\times \mathcal{D}$ the system is uniformly locally exponentially incrementally stable in this subset. If $\mathcal{D} = \mathbb{R}^n$ the definition is global.
\end{definition} 

The theory of exponential convergence of two trajectories of the same system has first been studied as \emph{uniform convergence} by Demidovich \cite{demidovich1967lectures,pavlov2004convergent}, and later extended through contraction theory \cite{lohmiller1998contraction, tsukamoto2021contraction}. Contraction is a sufficient condition for exponential incremental stability, that when $f$ is smooth, can be checked by the following condition.
\begin{theorem} 
    \label{the:demidovich}
    Let Assumption \ref*{assum:local_behavior}.\ref*{assum:local_diff} hold, and suppose that there exists a sequence of positive-definite, bounded matrices $P(k) \in \mathbb{R}^{n\times n}$, $k\in \mathbb{N}_{\geq k_0}$, and some $\rho\in(0,1)$ such that
    \begin{equation}
        \label{eq:demidovich}
        D(k,x): = \frac{\partial f}{\partial x}(k,x)^T P(k+1) \frac{\partial f}{\partial x}(k,x) -\rho^2 P(k) \prec 0
    \end{equation}
    for all $x \in \mathcal{\tilde{D}}_0\subseteq \mathcal{D}_0$. Then, there exists a constant $r_{\mathcal{\tilde{D}}_0}\in \mathbb{R}_+$, such that \eqref{eq:nonlinearsystem} is uniformly locally exponentially incrementally stable in $ \mathcal{B}(r_{\mathcal{\tilde{D}}_0})\subseteq \mathcal{\tilde{D}}_0$ with a rate $\rho$.
\end{theorem}

We also state the following converse Lyapunov theorem.
\begin{theorem}
    \label{the:converse_edelta_s}
    If the system \eqref{eq:nonlinearsystem} is uniformly  exponentially incrementally stable in $\mathcal{D}$, then there exists a function $V:  \mathbb{N}_{\geq k_0} \times \mathcal{D}\times \mathcal{D}\rightarrow \mathbb{R}_+$ and constants $c_1,c_2\in \mathbb{R}_+$ and $\beta \in (0,1)$, such that for all $(\xi,\zeta) \in \mathcal{D} \times \mathcal{D}$ and $k \in \mathbb{N}_{\geq k_0}$
    \begin{align}
        \label{eq:incremental_lyapunov_1}
        c_1 \|\xi-\zeta\|^2 \leq V(k,\xi,\zeta) \leq c_2 \|\xi-\zeta\|^2, \\
        \label{eq:incremental_lyapunov_2}
        V\left(k+1,f(k,\xi),f(k,\zeta)\right) \leq \beta^2 V(k,\xi,\zeta).
    \end{align}
    Moreover, under Assumption \ref*{assum:local_behavior}.\ref*{assum:local_lipschitz}, there exists a constant $c_3 \in \mathbb{R}_+$, such that for all $(\xi,\zeta,\tilde{\xi},\tilde{\zeta}) \in \mathcal{D}\times \mathcal{D}\times \mathcal{D}\times \mathcal{D}$ and $k \in \mathbb{N}_{\geq k_0}$
    \begin{equation}
    \label{eq:lyapunov_difference_incremental}
    \begin{split}
        &|V(k,\xi,\zeta)-V(k,\tilde{\xi},\tilde{\zeta})|\leq \\
        &\leq c_3 \left(\|\xi-\tilde{\xi}\| + \|\zeta-\tilde{\zeta}\|\right)\left(\|\xi - \zeta\|+\|\tilde{\xi}-\tilde{\zeta}\|\right).
    \end{split}
    \end{equation}
\end{theorem}

The proofs of Theorem \ref{the:demidovich}, adapted from \cite{tsukamoto2021contraction}, and Theorem \ref{the:converse_edelta_s}, extended from \cite{angeli2002lyapunov} and \cite{tran2018convergence}, are provided in the appendix. 

\subsection{Main Results}

In this subsection, we show that if the nonlinear dynamics \eqref{eq:nonlinearsystem} are uniformly exponentially stable in $\mathcal{D}$ and satisfy Assumption \ref*{assum:local_behavior}, then they are also \acrshort*{ediss} in the same region. First, we show that under the local Lipschitz continuity assumption, exponential incremental stability implies \acrshort*{ediss}.
\begin{theorem}
    \label{prop:eis-eiiss}
    Let Assumption \ref*{assum:local_behavior} hold, then if the nonlinear system \eqref{eq:nonlinearsystem} is uniformly exponentially incrementally stable in $\mathcal{D}$, it is \acrshort*{ediss} in the same region.
\end{theorem}

For completeness, we  provide the proof in the appendix, and note that the implication  is similar to existing results in the literature, e.g. \cite{jiang2001input}, that studies the asymptotic stability of the equilibrium, or \cite[Lemma 2]{kohler2022globally} that considers a non-autonomous setting.

Next, we show that if in addition to local Lipschitz continuity, the nonlinear dynamics are also locally differentiable in some arbitrarily small region $\mathcal{D}_0$ around the equilibrium, then uniform exponential stability implies uniform exponential incremental stability and \acrshort*{ediss}. The results are formalized in Theorem \ref{the:exp_implies_edis} and Corollary \ref*{cor:exp_implies_ediss}.

\begin{theorem}\label{the:exp_implies_edis}
    Let Assumption \ref*{assum:local_behavior} hold, then if  the nonlinear system \eqref{eq:nonlinearsystem} is  uniformly exponentially stable in $\mathcal{D}$, it is also uniformly exponentially incrementally stable in the same region. 
\end{theorem}
\begin{proof}
    We start by showing that the exponential stability of $f$ implies that the linearized dynamics around the origin are also stable by following similar arguments to \cite{khalil2002nonlinear}. Let
    \begin{equation*}
        A(k) : = \frac{\partial f(k,x)}{\partial x}(k,\boldsymbol{0}),
    \end{equation*}
    which is well-defined given Assumption \ref*{assum:local_behavior}. Moreover, there exists a $\bar{A} \in \mathbb{R}_+$, such that $\|{A}(k)\| \leq \bar{A}$, $k \in \mathbb{N}_{\geq k_0}$. It follows from Theorem \ref*{the:converse_lyapunov}, that there exists a continuous mapping $V:k \in \mathbb{N}_{\geq k_0} \times \mathcal{D} \rightarrow \mathbb{R}$ satisfying \eqref{eq:converse_lyapunov_1}-\eqref{eq:converse_lyapunov_2}. Let us consider $V$ as a candidate Lyapunov function for $A(k)$. Then for all $x \in \mathcal{D}$, $k\in \mathbb{N}_{\geq k_0}$ there exist constants $\beta \in (0,1)$, and $c_4,d \in \mathbb{R}_+$ such that
    \begin{align*}
        &V(k+1,A(k)x) =\\
        & V(k+1,f(k,x)) + \left[V(k+1,A(k)x)-V(k+1,f(k,x))\right]\\
        &\leq \beta^2 V(k,x) + \left[V(k+1,A(k)x)-V(k+1,f(k,x))\right]\\
        &\leq \beta^2 V(k,x) + c_4 \|f(k,x)- A(k)x\|\|f(k,x) + A(k)x\|\\
        &\leq \beta^2 V(k,x) +  c_4 \|x\|\cdot \|f(k,x)- A(k)x\|\left(d\beta  + \|A(k)\|\right),
    \end{align*}
    where the first inequality follows from Theorem \ref*{the:converse_lyapunov}, the second from Lemma \ref*{lem:delta_V_bound} and the last from Definition \ref*{def:exponential_stability} and properties of induced norms. Denoting $h(k,x): = f(k,x)- A(k)x$, it follows from the Lipschitz continuity of  the Jacobian of $f$ \cite[Chpt. 4.6]{khalil2002nonlinear} that there exists a $L_h \in \mathbb{R}_+$, such that $\|h(k,x)\| \leq L_h \|x\|^2$, for all $k\in \mathbb{N}_{\geq k_0}$ for all $x\in \mathcal{D}_0$.  Using this
    \begin{align*}
        &V(k+1,A(k)x) \leq \beta^2 V(k,x)  +c_4 L_h\left(d\beta + \bar{A}\right) \|x\|^3\\
                    &\leq \left(\beta^2+ \frac{c_4L_h\left(d \beta + \bar{A}\right)}{c_1}\|x\|\right) V(k,x): = \gamma V(k,x),
    \end{align*}
    where the second inequality follows from the converse Lyapunov Theorem \ref*{the:converse_lyapunov} for some $c_1 \in \mathbb{R}_+$. Defining $r_1: = \min \left\{\frac{c_1\left(1-\beta^2\right)}{c_4 L_h \left(d\beta+\bar{A}\right)}, \max\{r>0 \ |\ \|x\|<r, x \in \mathcal{D}_0\}\right\}$, for all $\|x\| < r_1$, it holds that $\gamma <1$.    
    Hence, using Theorem \ref*{the:direct_lyapunov} the linearized dynamics $x_{k+1} =A(k) x_k$ are uniformly exponentially stable. Then, from \cite[Thm 23.3]{rugh1996linear} there exists a sequence of uniformly positive definite $P(k)$ that solves the difference Lyapunov equation $ A^\top(k) P(k+1)A(k) - P(k) \leq - c I,$ for some $c \in \mathbb{R}_+$.

    Considering now equation \eqref{eq:demidovich} for some $k \in \mathbb{N}_{\geq k_0}$, $\bar{x} \in \mathcal{D}_0$ and denoting $d_k(\bar{x}): = \frac{\partial f}{\partial x}(k,\bar{x}) - A(k)$, it holds that
    \begin{align*}
        &\frac{\partial f}{\partial x}(k,\bar{x})^\top P(k+1) \frac{\partial f}{\partial x}(k,\bar{x}) - P(k) \\
        &= \left[A(k)+d_k(\bar{x})\right]^\top P(k+1) \left[A(k)+d_k(\bar{x})\right] -P(k)\\
        &= A(k)^\top P(k+1) A(k) - P(k)\\
         &+A(k)^\top P(k+1) d_k(\bar{x}) + d_k(\bar{x})^\top P(k+1) \left[A(k)+d_k(\bar{x})\right]\\
        &\leq -cI + A(k)^\top P(k+1)d_k(\bar{x})\\
         &\qquad \qquad+ d_k(\bar{x})^\top P(k+1) \left[A(k)+d_k(\bar{x})\right],
    \end{align*}
    where the first equality follows from the definition of $d_k(\bar{x})$ and the last one from \cite[Thm 23.3]{rugh1996linear}. 

    Following the same arguments as in \cite[Chpt. 4.6]{khalil2002nonlinear}, there exists a constant $L_2 \in \mathbb{R}_+$, such that for all $k \in \mathbb{N}_{\geq k_0}$ and $\bar{x} \in \mathcal{D}_0$, $\|d_k(\bar{x})\|\leq L_2 \|\bar{x}\|$. Pre- and post-multiplying  the above with some $x^\top$ and $x$, respectively then yields
    \begin{align*}
        &x^\top \left[\frac{\partial f}{\partial x}(k,\bar{x})^\top P(k+1) \frac{\partial f}{\partial x}(k,\bar{x}) - P(k)\right]x \leq \\
         &-c\|x\|^2+ 2\|x\|^2\|A(k)^\top P(k+1) d_k(\bar{x})\|\\
         &\qquad \qquad + \|x\|^2\|d_k(\bar{x})^\top P(k+1) d_k(\bar{x})\| \\
        &\leq -\|x\|^2\left(c-2\bar{A}\bar{P}L_2 \|x\| - \bar{P}r_dL_2^2 \|x\|\right), 
    \end{align*}
    where $r_d = \max\limits_{x\in \mathcal{D}} \|x\|$, and $\|P(k)\| \leq \bar{P}, \; k \in \mathbb{N}_{\geq k_0}$. Note that the rate of exponential stability of the linear system $A(k)$ is $\sqrt{1 - \frac{c}{\bar{P}}} \in (0,1)$. Then, choosing $\rho^2> 1 - \frac{c}{\bar{P}}$, adding $x^\top\left(1-\rho^2\right) P(k)x$ to both sides of the above inequality and defining $r_2 :=  \frac{c - \left(1-\rho^2\right)\bar{P}}{2\bar{A}\bar{P}L_2 + \bar{P}r_dL_2^2}$ ensures that
    \begin{equation*}
        x^\top \left[\frac{\partial f}{\partial x}(k,x)^T P(k+1) \frac{\partial f}{\partial x}(k,x) -\rho^2 P(k)\right]x <0,
    \end{equation*}
    uniformly in $k$ and $x$, for all $\|x\| < r:= \min\left(r_1,r_2\right)$. This implies by Theorem \ref*{the:demidovich} that there exists a positive constant $r_{\mathcal{\tilde{D}}_0} \leq r$, such that for all $x \in \mathcal{B}(r_{\mathcal{\tilde{D}}_0})$ the system \eqref{eq:nonlinearsystem} is uniformly exponentially incrementally stable with rate of $\rho$.

    To show that the system is also uniformly exponentially incrementally stable in $\mathcal{D}$, consider any $\xi_1,\xi_2 \in \mathcal{D}$, then for all $k \in \mathbb{N}_{\geq k_0}$, and some $d\in \mathbb{R}_+$, the following two inequalities hold
        \begin{align}
        \label{eq:exp_bound}
        \begin{split}
            &\|\phi(k,k_0,\xi_1)-\phi(k,k_0,\xi_2)\| \leq \\
             &\qquad \|\phi(k,k_0,\xi_1)\|+ \|\phi(k,k_0,\xi_2)\| \leq 2r_d d\beta^{k-k_0},
        \end{split}\\
            \label{eq:lipschitz_bound}
            &\|\phi(k,k_0,\xi_1)-\phi(k,k_0,\xi_2)\| \leq L_f^{k-k_0}\underbrace{\|\xi_1-\xi_2\|}_{:=\Delta \xi}.
        \end{align}
        The bound in \eqref{eq:exp_bound} follows from the exponential stability of $f$, and the one in \eqref{eq:lipschitz_bound} from its Lipschitz continuity. Combining the two
        \begin{equation*}
            \|\phi(k,k_0,\xi_1)-\phi(k,k_0,\xi_2)\|\leq \min\{2r_d d \cdot\beta^{k-k_0}, L_f^{k-k_0}\|\Delta \xi\|\}.
        \end{equation*}
    Define $k'\in \mathbb{N}_{\geq k_0}$ such that both $\|\phi(k',k_0,\xi_1)\| < r_{\mathcal{\tilde{D}}_0},\|\phi(k',k_0,\xi_2)\|<r_{\mathcal{\tilde{D}}_0}$. Then, from the above analysis, there exists a $d'\in \mathbb{R}_+$, such that for all $k\in \mathbb{N}_{\geq k'}$
    \begin{align*}
        &\|\phi(k,k_0,\xi_1)-\phi(k,k_0,\xi_2)\|\\
        & \leq d'\rho^{k-k'}\cdot\|\phi(k', k_0,\xi_1)-\phi(k',k_0,\xi_2)\|.
    \end{align*}
    It then follows that
    \begin{align*}
        \|\phi(k,k_0,\xi_1)&-\phi(k,k_0,\xi_2)\| \\
        &\leq d'\rho^{k-k'}\cdot\min\{r_d d \cdot\beta^{k-k_0},L_f^{k-k_0}\|\Delta \xi\|\}.
    \end{align*}
    Note that for all $k=\{k_0, \hdots, k'\}$
    \begin{equation*}
        \|\phi(k,k_0,\xi_1)-\phi(k, k_0,\xi_2)\| \leq c_5 \|\Delta \xi\| \rho^{k-k_0},
    \end{equation*} 
    where $c_5:= \max\{1, \frac{L_f^{k'}}{\rho^{k'}}\}$, makes it a constant independent of $\|\Delta \xi\|$, and since $k'$ is finite, also independent of time. Combining the bounds, the following holds for all $k\in \mathbb{N}_{\geq k_0}$
    \begin{equation*}
        \|\phi(k,k_0,\xi_1)-\phi(k,k_0,\xi_2)\| \leq c_5d'\rho^{k-k_0} \|\Delta \xi\|,
    \end{equation*}
    which is the definition of uniform exponential incremental stability. 
\end{proof}
Combining Theorems \ref*{prop:eis-eiiss} and \ref*{the:exp_implies_edis} the following corollary follows directly.
\begin{corollary}
    \label{cor:exp_implies_ediss}
    Let Assumption \ref*{assum:local_behavior} hold, then if  the nonlinear system \eqref{eq:nonlinearsystem} is  uniformly exponentially stable in $\mathcal{D}$,  it is also \acrshort*{ediss} in the same region. 
\end{corollary}

\begin{remark}
    The conditions on the Jacobian matrix of $f$ in Assumption \ref*{assum:local_behavior} are required only in the time-varying case, as these are satisfied trivially in the time-invariant setting; for further details, we refer the reader to \cite[Theorems 4.7, 4.12]{khalil2002nonlinear}.
\end{remark}

\begin{remark}
    It is worth highlighting that the implication  in Corollary \ref{cor:exp_implies_ediss} is also claimed in \cite[Thm. 12]{tran2018convergence} without assuming local differentiability. While the proof in \cite{tran2018convergence} holds for asymptotic stability, the same arguments break down in the exponential case since one cannot find constants $d\in \mathbb{R}_+$ and $\lambda\in (0,1)$ as in Definition \ref{def:incremental_stability} that are independent of the initial deviation $\|\xi-\zeta\|$. Hence, the results in this section serve to fill this gap by deriving sufficient conditions in the form of Assumption \ref{assum:local_behavior} under which the implication holds.
 \end{remark}

In the sequel, we use these insights to address the closed-loop dynamics under suboptimal and optimal linear MPC policies as a notable use case.

%% file: S5_usecases.tex
\section{Model Predictive Control - A Use Case}
\label{sec:use_cases}

In Section \ref*{sec:suboptimality}, we showed that under certain assumptions on the suboptimal policy $\mu$ (Assumption \ref*{assum:mu_contractive}) and the benchmark policy $\mu^\star$ (Assumption \ref*{assum:mu_star}), the suboptimality gap of $\mu$ can be bounded for a certain family of costs. We now exploit the results of Section IV to show that these assumptions are satisfied in the case of suboptimal, PGM-based linear-quadratic MPC, and derive an asymptotically tighter suboptimality gap bound compared to \cite{karapetyan2023finite} under the same assumptions.

Consider the control of linear time-invariant dynamical systems, modeled by 
\begin{equation*}
    x_{k+1} = Ax_k+Bu_k,\quad  k\in \mathbb{N},
\end{equation*}
where $A\in \mathbb{R}^{n\times n}$, and $B\in\mathbb{R}^{n\times m}$ are known system matrices, and the state and input vectors are defined as before. We consider the finite horizon linear-quadratic regulator (LQR) problem with the objective of minimizing the finite-time cost
\begin{equation}
    \label{eq:quadratic}
    J_T(x_0,\boldsymbol{u}) = \|x_T\|^{2}_P+\sum_{k=0}^{T-1}\|x_k\|^{2}_Q+ \|u_k\|^2_R,
\end{equation}
where $Q\in \mathbb{R}^{n \times n}$ and $R \in \mathbb{R}^{m\times m}$ are design matrices and $P$ is taken to be the solution of the discrete-time Algebraic Riccati Equation, $P = Q + K^\top R K + (A-BK)^\top P (A-BK)$, with $K = (R+B^\top PB)^{-1}(B^\top PA)$. The control inputs must satisfy
$u_k \in \mathcal{U}$ for all $k\in \mathbb{N}$, where  $\mathcal{U} \subseteq \mathbb{R}^m$ is a constraint set. The following standard assumptions ensure a unique minimizer for \eqref{eq:quadratic} always exists \cite{mayne2000constrained,bemporad2002explicit}.

\begin{assumption}
    (Well-posed problem)   
\label{assum:well_posed}
\begin{enumerate}[label=\roman*.]
    \item The pair $(A,B)$ is stabilizable, $Q\succ 0$, $R\succ 0$.
    \item The input constraint set $\mathcal{U}$ is closed, convex and contains the origin in its interior. 
    \label{assum:convex_u}
\end{enumerate}
\end{assumption}

The model predictive controller solves this problem in a receding horizon fashion, solving the following parametric optimal control problem (POCP) at each timestep $k$, having measured a state $x \in \mathbb{R}^n$
\begin{equation}
    \label{eq:MPC_POCP}
    \begin{split}
        \mu^{\star}(x) := &\argmin_{\nu} \;  J_N(\xi_0,\nu)\\
         \text{s.t.} \; & \xi_{i+1} = A\xi_i +B\nu_i, \; i \!= \!0,\dots, N-1,\\
        &\xi_0 = x, \; \nu_i \in \mathcal{U}, \; i \!= \!0,\dots, N-1.
    \end{split}
     \end{equation}
Here $N$ is the prediction horizon length, and  $\nu = [{\nu}_0^{\top} \hdots {\nu}_{N-1}^{\top}]^{\top}$ denotes the predicted input vector. We refer to the minimiser of \eqref{eq:MPC_POCP} for a given initial state (parameter) $x\in\mathbb{R}^n$, $\mu^{\star}(x):\mathbb{R}^n \rightarrow \mathbb{R}^{Nm}$, as the optimal mapping. For this example, we take as the benchmark policy $\mu^\star$, the map solving the POCP \eqref{eq:MPC_POCP}. The optimal cost attained by this mapping is denoted by $J_N^\star(x):= J_N(x,\mu^{\star}(x))$, which serves as an approximate value function for the problem. For each $k$, the model predictive controller applies the first element of $\mu^{\star}(x_k)$ to the system, and the process is repeated in a receding horizon fashion. The optimal state evolution under this optimal MPC policy is then given by
\begin{equation}\label{eq:lqmpc_optimal}
    x^{\star}_{k+1} = Ax_k^{\star} +\overline{B}\mu^{\star}(x_k^{\star}):= f(x^{\star}_k), \quad k\in \mathbb{N},
\end{equation}
where $x_0^{\star}:=x_0$, $\overline{B} := BS$,  and $S := \left[I_{m\times m}~\boldsymbol{0}~\hdots~\boldsymbol{0}\right] \in \mathbb{R}^{m \times Nm}$ is a matrix selecting the first control input. Note that the optimal MPC policy, $\mu^\star$ is time-invariant due to the structure of the problem. 

Problem \eqref{eq:MPC_POCP} is a parametric quadratic program and for a given parameter $x\in \mathbb{R}^n$ can be represented in an equivalent condensed form $J_N^\star(x) =\min_{\nu\in \mathcal{N}}  \|(x,\nu)\|_M^2$, where  $\mathcal{N} = \mathcal{U}^N \subseteq \mathbb{R}^{Nm}$
\begin{equation}
\label{eq:condensed_matrices}
    M = \begin{bmatrix}
W & G^{\top}\\
G & H
\end{bmatrix},
\end{equation}
and the definitions of $H\in \mathbb{R}^{Nm\times Nm}$, $W\in \mathbb{R}^{n \times n}$ and,  $G \in \mathbb{R}^{Nm \times n}$ can be found in \cite{liao2021analysis}.

As the optimal $\mu^\star(x)$ may often be infeasible to compute exactly, suboptimal schemes are often considered. In our setting, a suboptimal policy is computed by performing only a finite number of optimization steps for \eqref{eq:MPC_POCP}. In particular,  given $x\in \mathbb{R}^n$ and an input vector ${\nu} \in \mathbb{R}^{Nm}$, consider the operator that performs one step of the projected gradient method
\begin{equation}
\label{eq:pgm_operator}
    \mathcal{T}(x,\nu):= \Pi_{\mathcal{N}}[\nu - \alpha \nabla{\nu}{J_N}(x,\nu) ],
\end{equation}
where $\alpha\in\mathbb{R}$ is a step size. Applying \eqref{eq:pgm_operator} iteratively  $\ell_k \in \mathbb{N}_+$ times provides an approximation for the optimal input, and hence the optimal policy. The combined dynamics of the system and the approximate optimizer are then given by\footnote{The subscript of $\ell_k$ is dropped when it is taken to be a constant.}
\begin{subequations}
\label{eq:suboptimal_combined}
\begin{align}
  \label{eq:suboptimal_input}
  u_{k}^\mu &= \mathcal{T}^{\ell_k}(x_{k},u_{k-1}^\mu),\\
    x_{k+1} & = Ax_k + \overline{B}u_{k}^\mu,
    \label{eq:suboptimal_system}
\end{align}
\end{subequations}
where $u_{-1}^\mu \in \mathbb{R}^{Nm}$ is an initialization vector, and for some $l \in \mathbb{N}, x\in\mathbb{R}^n$ and $\nu\in\mathbb{R}^{Nm}$, we define
\begin{equation}
    \label{eq:composition}
    \mathcal{T}^{l}(x,\nu) = \mathcal{T}(x,\mathcal{T}^{l-1}(x,\nu)), \qquad \mathcal{T}^0(x,\nu) = \nu.
\end{equation}
The dynamics under the suboptimal policy are \eqref{eq:suboptimal_system} by taking $u_{k}^\mu =  \mu_k (x_k)$ for all $k\in \mathbb{N}$, i.e.
\begin{equation*}
    x_{k+1}  = \underbrace{Ax_k + \overline{B}{\mu}^{\star} (x_k)}_{f(x_k)} + \overline{B}\underbrace{\left(u_{k}^\mu - \mu^{\star}(x_k)\right)}_{=d_k}, \quad k\in \mathbb{N}.
\end{equation*}
Note that $\mu_k$ is also a function of the previous input state $u_{k-1}^\mu$. However, since the closed-loop evolution is noise-free, it can be uniquely determined given the initialization vector, the current time $k$, and the current state. Hence, the dependence on $u_{k-1}^\mu$ is encoded in the subscript of $\mu_k$. The suboptimal policy can in general be defined as a function of the information vector $\mathcal{I}_k = \{x_k, u_{k-1}, \hdots, u_0\}$; as long as Assumption \ref*{assum:mu_contractive} is satisfied, the results in this manuscript hold.
\subsection{Optimal MPC}
\label{sec:optimal_mpc}
In this subsection, we review the properties of the optimal mapping $\mu^{\star}(x)$. As shown in \cite{liao2021analysis, leung2021computable,limon2006stability}, the system \eqref{eq:lqmpc_optimal} is exponentially stable with the forward invariant ROA estimate
\begin{equation*}
    \Gamma_N: = \{x\in\mathbb{R}^n \mid \psi(x)\leq r_N\},
\end{equation*}
where $\psi(x):=\sqrt{J_N^\star(x)}$, $\textstyle{d = c\cdot {\lambda^{-}(Q)}/{\lambda^{+}(P)}}$, $r_N = \sqrt{N d + c}$ and $c>0$ is such that the following set
\begin{equation*}
    \Omega = \{x\in\mathbb{R}^n \mid \|x\|_P^2\leq c \},
\end{equation*}
also satisfies $\Omega \subset \{x\in\mathbb{R}^n \mid -Kx \in \mathcal{U}\}$. The function $\psi(x)$ is a Lyapunov function for the optimal MPC algorithm, satisfying
\begin{align}
    \label{eq:V_bounds}
    \|x\|_P \leq \psi(x) \leq \|x\|_W\\
    \label{eq:exp_rate_optimal_v}
    \psi\left(f(x)\right) \leq \beta \psi(x),
\end{align}
where $\beta\in (0,1)$ is the exponential decay rate.
The Lipschitz continuity of the optimal mapping is formalized in the following lemma. 
\begin{lemma} \cite[Corollary~2]{liao2021analysis}
    \label{lem:lispchitz}
        For any $(x,y)\in \Gamma_N \times \Gamma_N$, the optimal solution mapping, $\mu^{\star}(x)$, satisfies
        \begin{align*}
            &\|\mu^{\star}(x)-\mu^{\star}(y)\|\leq \|H^{-\frac{1}{2}}\|\|G(x-y)\|_{H^{-1}} \leq L \|x-y\|
        \end{align*}
        with a Lipschitz constant $L: = \|H^{-\frac{1}{2}}\|\cdot \|H^{-\frac{1}{2}}G\|$.
    \end{lemma}
The proof follows from the parametric quadratic program structure of the MPC problem and is derived in \cite{liao2021analysis} or \cite{bemporad2002explicit} from an explicit MPC point of view.

\subsection{Suboptimal MPC}

The suboptimal policy in this setting is defined by \eqref{eq:suboptimal_input}. The following well-established result shows the linear rate of convergence  of the PGM method.
\begin{theorem}\cite[Theorem~3.1]{taylor2018exact} 
    \label{the:pgm_contraction}
    For any $x \in \mathbb{R}^n$, $\nu \in \mathbb{R}^{N m}$, $\ell \in \mathbb{N}_+$, and for $\alpha = \frac{1}{\lambda^{+}(H)+\lambda^{-}(H)}$
    
    $$
    \left\|\mathcal{T}^{\ell}(x,\nu)-\mu^{\star}(x)\right\| \leq \eta^{\ell}\|\nu-\mu^{\star}(x)\|,
    $$
    where $\eta=(\lambda^{+}(H) -\lambda^{-}(H)) /(\lambda^{+}(H)+\lambda^{-}(H))$.
\end{theorem}
The suboptimal MPC scheme is treated by considering the combined evolution of the system-optimizer dynamics \eqref{eq:suboptimal_combined}. The stability of such a scheme, also referred to as Time-Distributed MPC (TD-MPC) or as real-time implementation of MPC is shown in \cite{zanelli2020lyapunov,liao2020time,liao2021analysis} for a fixed number of iterations $\ell$ and in \cite{karapetyan2023finite} for a time-varying $\ell_k$. In particular, if $\ell_k>\ell^{\star}$ for all $k\in \mathbb{N}$, where 
\begin{equation*}
    \ell^{\star} = \frac{\log(1-\beta) - \log(\sigma\kappa + \omega(1-\beta))}{\log(\eta)},
\end{equation*}
where $\beta$ is the same as in Section \ref{sec:optimal_mpc},  $\omega = 1 + \|H^{-\frac{1}{2}}\|\|H^{-\frac{1}{2}}G\overline{B}\|$,  $\sigma = \|W^{\frac{1}{2}}\overline{B}\|$, and
\begin{align*}
\kappa &=\|H^{-\frac{1}{2}}\|\|H^{-\frac{1}{2}}G(A-I)P^{-\frac{1}{2}}\|\\
&+\|H^{-\frac{1}{2}}\|\sqrt{\lambda_H^+(G\overline{B})(\lambda_P^+(W)-1)},
\end{align*}
Then, the combined dynamics \eqref{eq:suboptimal_combined} are exponentially stable in the following forward invariant ROA estimate
\begin{align*}
    \Sigma_N = \biggl\{ (x, z)\! \in\! \Gamma_N\! \times\! \mathcal{N} \mid~& \psi(x)\! \leq\! r_N, \biggr. \\ 
        &\biggl. \|z-\mu^{\star}(x)\| \leq \frac{(1-\beta)r_{N}}{\sigma} \biggl\},
\end{align*} 
recalling $r_N$ from Section \ref{sec:optimal_mpc}. For the rest of the section we take $u_{-1}^\mu=0$, with the assumption that the resulting $u_{0}^\mu = \mathcal{T}^\ell(x_0,u_{-1}^\mu)$ is such that $(x_0,u_{0}^\mu)\in \Sigma_N$. For further analysis of the choice of the initial input and the forward invariant dynamics, we refer the interested reader to \cite{leung2021computable}.

The exponential stability result from \cite{karapetyan2023finite} is summarized in the following Lemma.
\begin{lemma}\label{lem:x_upper_bound_varying_l}
    \cite[Lemma 5]{karapetyan2023finite}
    Given the dynamics \eqref{eq:suboptimal_combined}, the following holds for  all $(x_0,u_{0}^\mu)\in \Sigma_N$, $k\in \mathbb{N}$ and  $\ell_k> \ell^{\star}$
\begin{equation*}
    \|x_k\| \leq   h_0\|P^{-\frac{1}{2}}\|\cdot\|x_0\|_W\prod_{i=-1}^k\varepsilon_i,
\end{equation*}
where $\varepsilon_k := \max\{\beta+\tau\kappa\eta^{\ell_k}, \frac{\sigma+\tau\eta^{\ell_k}\omega}{\tau}\} \in (0,1)$, $\varepsilon_{-1}:=1$ and $h_0 = 1+\tau\eta^{\ell_0}L\|W^{-\frac{1}{2}}\|$.
\end{lemma}
If the same number of optimization iterations are taken at all times the bound reduces to the following.
\begin{corollary}\label{lem:x_upper_bound}
    \cite[Corollary 3]{karapetyan2023finite} Given the dynamics \eqref{eq:suboptimal_combined}, the following holds for  all  $\ell> \ell^{\star}$, $(x_0,u_{0}^\mu)\in \Sigma_N$ and  $k\in \mathbb{N}$
   \begin{equation*}
       \|x_k\| \leq h\|P^{-\frac{1}{2}}\|\cdot\|x_0\|_W \cdot  \varepsilon^k,
   \end{equation*}
   where  $\varepsilon := \max\{\beta+\tau\kappa\eta^{\ell}, \frac{\sigma+\tau\eta^{\ell}\omega}{\tau}\} \in (0,1)$ and $h = 1+\tau\eta^{\ell}L\|W^{-\frac{1}{2}}\|$.
\end{corollary}

\subsection{Suboptimality Gap}

We define the suboptimality vector, ${\tilde{\eta}}_\ell$, for this use case as
    \begin{equation*}
         {\tilde{\eta}}_\ell: = 
         \begin{bmatrix}
            \tilde{\eta}_{\ell,1} & \tilde{\eta}_{\ell,2} & \hdots & \tilde{\eta}_{\ell,T-1}
        \end{bmatrix},
    \end{equation*} 
where $\tilde{\eta}_{\ell,k} := \eta^{\ell_k}(1+\tilde{\eta}_{\ell,k+1}),\; k\in [0,T-2]$, and $\tilde{\eta}_{\ell,T-1} := \eta^{\ell_{T-1}}$. We denote the Euclidean norm of the suboptimality vector by $\bar{\eta}_{\ell}:= \|{\tilde{\eta}}_\ell\|$. The main result for a suboptimal MPC scheme is summarized in the following Theorem.

 \begin{theorem}
    \label{the:lqmpc_tv}Let Assumption \ref{assum:well_posed} hold, then if $\ell_k> \ell^{\star}, \ \forall k\in \mathbb{N}$, the suboptimality gap of the suboptimal MPC \eqref{eq:suboptimal_combined} is given by
    \begin{equation*}
        \mathcal{R}_T(x_0) = \mathcal{O}\left( \bar{\eta}_\ell  \|x_0\|\right),
    \end{equation*}
    for all $(x_0,u_{0}^\mu) \in \Sigma_N$. Specifically, it is bounded by
    \begin{equation*}
        \mathcal{R}_T(x_0) \leq \bar{M}\left(\tilde{\eta}_{\ell,0}\|\delta u_0\|+ L{\Delta}^\top{\tilde{\eta}_{\ell}} \right),
    \end{equation*}
    with $\bar{M}, \Delta$ defined as in Theorem \ref*{the:incurred_suboptimality} and $\delta u_0 = u_{-1}^{\mu} - \mu_0^\star(x_0)$. Moreover, if  $\ell_k=\ell>\ell^\star, \forall k\in \mathbb{N}$, then, for all $(x_0,u_{0}^\mu) \in \Sigma_N$
    \begin{equation*}
        \mathcal{R}_T(x_0) = \mathcal{O}\left(\frac{\eta^\ell}{1-\eta^\ell} \|x_0\|\right).
    \end{equation*}    

\end{theorem}
\begin{proof}
    We start by showing that Assumptions \ref{assum:mu_star}-\ref{assum:stage_cost} are satisfied in the MPC use-case. In this setting, the linear dynamics to be controlled are given by
    \begin{equation}\label{eq:lqmpc_modified}
        x_{k+1} = Ax_k +\overline{B} \bar{u}_k, \quad k \in \mathbb{N}, 
    \end{equation}
    with the input $\bar{u}_k \in \mathbb{R}^{Nm}$. First, we note that Assumption \ref*{assum:Lipschitz_g} is satisfied trivially with $L_u=\|\overline{B}\|$. The benchmark controller is the optimal policy $\mu^\star$ that solves the POCP \eqref{eq:MPC_POCP}, and is given in \eqref{eq:lqmpc_optimal}. 

    Taking $\mathcal{D}^\star$ to be the forward invariant \acrshort*{roa} estimate $\mathcal{D}^\star = \Gamma_N $, it follows from Lemma \ref{lem:lispchitz} that Assumption \ref{assum:mu_star}.\ref{assum:Lipschitz} is satisfied. Since the policy is time-invariant, Assumption \ref*{assum:mu_star}.\ref*{assum:slow} is satisfied trivially with $a_k=0$ for all $k$. To show that the  Assumption \ref{assum:mu_star}.\ref{assum:ediss}  is  satisfied, we use the analysis from Section \ref{sec:sufficient_conditions}. Exploiting the structure of the optimization problem \eqref{eq:MPC_POCP}, it has been shown in \cite{bemporad2002explicit} that  under Assumption \ref{assum:well_posed}, the solution of the MPC problem is piecewise-affine in the state. Moreover, the set $\Omega$, that defines a forward invariant set around the origin for the dynamics \eqref{eq:lqmpc_optimal} with inactive input constraints, is non-empty \cite{bertsekas2012dynamic}. This and Lemma \ref{lem:lispchitz} imply that Assumption \ref{assum:local_behavior} is satisfied. As discussed in Section \ref*{sec:optimal_mpc}, the benchmark dynamics \eqref{eq:lqmpc_optimal} are exponentially stable in $\mathcal{D}^\star$, therefore, by Corollary \ref{cor:exp_implies_ediss} we conclude that they are also \acrshort*{ediss} in the same region.  The suboptimal policy is given by \eqref{eq:suboptimal_input}. The results in \cite{karapetyan2023finite} show that for all $\ell_k \geq \ell^\star$, $\Sigma_N$ is a forward invariant ROA estimate for the combined dynamics \eqref{eq:suboptimal_combined}. Given this and an initial $u_{-1}^\mu =0$, consider $\mathcal{D}^\mu =\{x\in \mathbb{R}^n \ | \ (x,\mathcal{T}^\ell(x,u_{-1}^\mu)) \in \Sigma_N\}$. Then Assumption \ref{assum:mu_contractive} is satisfied directly from Theorem \ref{the:pgm_contraction} and the suboptimal policy definition.

    Finally, to reconcile the quadratic cost defined in \eqref{eq:quadratic} and the modified dynamics \eqref{eq:lqmpc_modified}, we redefine the cost, as
    \begin{equation*}
        J_T(x_0, \bar{\boldsymbol{u}}) = \|x_T\|^{2}_P+\sum_{k=0}^{T-1}\|x_k\|^{2}_Q+ \|\bar{u}_k\|^2_{\overline{R}},
    \end{equation*}
    where $\overline{R} = S^\top RS$. For the quadratic costs in \eqref{eq:quadratic}, and for any $(x,y) \in \mathcal{D}^\mu \times \mathcal{D}^\mu$, and $(u,z) \in (\mathcal{N}\times \mathcal{N})$
    \begin{align*}
        |x^\top Qx - y^\top Qy &+ u^\top R u - z^\top R z|\\
        \leq & 2\|x-y\|\|Q\|x_m + 2 \|u-z\| \|R\|u_m,
    \end{align*}
    where $x_m$ and $u_m$ are such that, $\|x\|\leq x_m$ and $\|u\|\leq u_m$ for all $x\in \mathcal{D}^\mu, u\in \mathcal{N}$. Then the condition in Assumption \ref{assum:stage_cost} is satisfied with $M_x= 2x_m  \max\{\|Q\|,\|P\|\}$ and $M_u = 2u_m \|\overline{R}\|$. 

    Since Assumption \ref{assum:well_posed} implies Assumptions \ref{assum:mu_star}-\ref{assum:stage_cost} are satisfied, we can invoke the bound in Theorem \ref{the:incurred_suboptimality} for the suboptimality  gap. As the suboptimal dynamics are exponentially stable, its pathlength is finite. In particular
    \begin{align*}
        \mathcal{S}_{T,2}^2 &= \sum_{k=1}^{T}\|x_k-x_{k-1}\|^2 \leq 4\sum_{k=0}^{T}\|x_k\|^2\\
        &\leq 
        c_0^2 \|x_0\|^2\sum_{k=0}^{T}\prod_{i=0}^{k}\varepsilon^2_{i-1} \leq \frac{c_0^2}{1-\bar{\varepsilon}^2} \|x_0\|^2,
    \end{align*}
    where the first inequality follows by the triangle inequality, the second from Lemma \ref{lem:x_upper_bound_varying_l} and by denoting $c_0 : = 2 h_0\|P^{-\frac{1}{2}}\|\|W^{\frac{1}{2}}\|$, and the last one by bounding the geometric series and  denoting $\bar{\varepsilon}: = \max\{\varepsilon_i\}_{i=0}^{T-1}$. Noting that $a_k=0$ for all $k\in \mathbb{N}$ in this example, the suboptimality gap is  bounded by $\mathcal{O}(\mathcal{S}_{T,2}\bar{\eta}_\ell) = \mathcal{O}(\bar{\eta}_\ell \|x_0\|)$.

    In the case when $\ell_k = \ell>\ell^\star$ for all $k\in \mathbb{N}$, one can use the bound derived in Corollary \ref{cor:incremental_stability_corollary}, with a simplified expression for the pathlength
    \begin{equation*}
        \mathcal{S}_T \leq \frac{c}{1-\varepsilon}\|x_0\|.
    \end{equation*}
    The above is obtained directly from the bound in Corollary \ref{lem:x_upper_bound} and by denoting $c:  = 2 h\|P^{-\frac{1}{2}}\|\|W^{\frac{1}{2}}\|$. 
\end{proof}

The theorem shows that the suboptimality gap of the MPC suboptimal controller \eqref{eq:suboptimal_combined} scales with $\tilde{\eta}_\ell$ or $\eta^\ell$, where the number of iterations $\ell_k$ are design parameters. Note that the higher $\ell_k$ the more computation is required at each timestep, but the lower the suboptimality gap; in the limit, as $\ell \rightarrow \infty$, the suboptimality gap is zero. This is a tighter result than the one derived in \cite{karapetyan2023finite}, as in the latter no incremental properties of the optimal controller are used. Specifically, when looking at the limit case of $\eta_k = 0, \; k\in \mathbb{N}$, the suboptimality gap in \cite{karapetyan2023finite} is strictly positive, while the bound in Theorem \ref*{the:lqmpc_tv} vanishes, reflecting the exact matching of the suboptimal and benchmark trajectories. The derived bounds can be used by control designers to give a quantifiable measure of the finite-time suboptimality of the controller. This can then be utilized to find the best sequence of $\ell_k$ to deliver a desired tradeoff between suboptimality and computational power.

In practice, the suboptimal MPC can be asymptotically stable even when the number of optimization iterations $\ell$ is less than $\ell^\star$. In this case, the existence of a forward invariant region of attraction $\mathcal{D}^\mu$ is not given, but the bounds in Theorem \ref*{the:incurred_suboptimality} and Corollary \ref*{cor:incremental_stability_corollary} still hold, as long as the closed-loop suboptimal system stays stable. This is shown in the following numerical example.

\subsection{Numerical Example}
\label{sec:numerical}
The suboptimal TD-MPC scheme described in this section is implemented for  the following linearized, continuous-time model  of an inverted pendulum from \cite{leung2021computable}, \cite{karapetyan2023finite}
\begin{equation*}
    A_c=\left[\begin{array}{cc}
0 & 1  \\
14.7 & 0 \\
\end{array}\right],\: B_c =\left[\begin{array}{c}
0 \\
30
\end{array}\right],
\end{equation*}
where the state is $x= [\theta,~\dot{\theta}]^\top$, $\theta$ is the angle relative to the equilibrium position and the control input is the applied torque. We consider the control of the discretized model of the plant with a sampling time of $T_s=0.1$. The input constraint set is taken to be $\mathcal{U}= [-1,1]$, the cost matrices are  $Q =  I_2$, and $R = 1$ and the initial state is $x_0 = [-\pi/4~\pi/3]^{\top}$.

\begin{figure}%
    \begin{center}
    \includegraphics[width=\columnwidth]{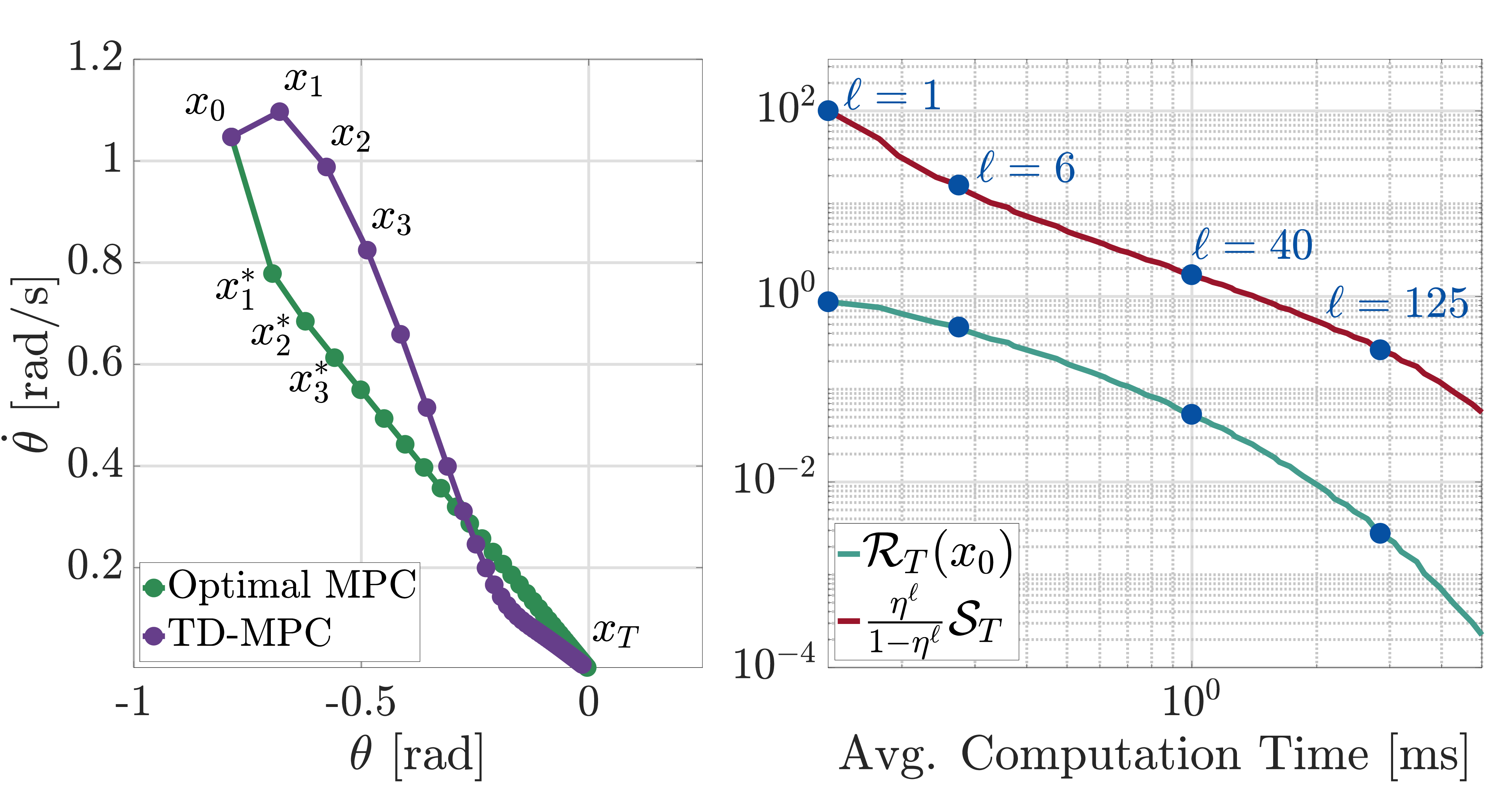}   
    \caption{The phase plot on the left shows two separate trajectories generated by applying respectively TD-MPC, with a constant $\ell=6$ (in green) and optimal MPC policies (in purple). The logarithmic scale plot on the right shows the empirical suboptimality gap (in green), as well as the order of the upper bound decrease (in red) as $\ell$, and therefore the computation time is increased. The two curves on the right are parameterized by $\ell$, ranging from $1$ to $5000$.}
    \label{fig:mpc_trajectories}
    \end{center}
\end{figure}

The left panel of Figure \ref*{fig:mpc_trajectories} shows the evolution of two trajectories in closed-loop with the TD-MPC policy with $\ell=6$ and with an optimal MPC. For this example $\ell^\star=849$. However, even with the  low value of $\ell=6$ the closed-loop system stays stable, as also observed in \cite{liao2021analysis,leung2021computable}. Although the asymptotic/exponential stability cannot be proven, the finite time bounds can still be computed online using only the suboptimal states as per Corollary \ref*{cor:incremental_stability_corollary}. The order of this upper bound, as well as the empirically observed suboptimality gap, $\mathcal{R}_T(x_0)$, are plotted in the right panel of Figure \ref*{fig:mpc_trajectories} for $T=30$ for a range of values of $\ell$, increasing from $1$ to $5000$. The decrease of the suboptimality gap for increasing values of $\ell$ is juxtaposed with the increase of simulation/computational time in the same figure. The simulation time for each $\ell$ is calculated as the sum of the times it takes to solve the TD-MPC for each timestep, over the horizon $T$.  To obtain an averaged value for this time, its average over $100$ repeated independent runs from the same initial conditions is taken. The initial states are intialized in $\Sigma_N$ following the procedure described in \cite{leung2021computable}. In the right panel of the figure, only the complexity $\frac{\eta^\ell}{1-\eta^\ell}\mathcal{S}_T$ of the suboptimality gap upper bound is plotted. The true constant multiplying the complexity term above is much larger; our aim here is not to compare the very conservative theoretical bound with the practical performance, but to give an estimate of whether the bound captures the order correctly. Indeed, it can be observed from the figure that the order of the true suboptimality gap is approximately captured by the upper bound with an underestimation.

Among other possible uses of the closed-loop suboptimality analysis, the insights in the figure can be used to design the allocation of finite computational resources. The right-side plot in the figure can be used to estimate the relative gain in computational time and loss in optimality for a given change in $\ell$. For example, a change of $\ell=6$ to $\ell=40$, or equivalently $\eta^\ell = 0.92$ to $\eta^\ell=0.56$ results in a $3.6$ times increase in computational time and a $9.3$ times decrease in the suboptimality gap bound. This provides an approximation for the variation of the true subopimality gap that decreased $8.7$ times in the same interval.

%% file: S6_conclusions.tex
\section{Conclusions}

We study the finite-time suboptimality gap of policies for discrete-time nonlinear, time-varying systems. We show that when the benchmark policy is chosen to be exponentially incrementally stable, then given a geometric improvement condition on the suboptimal policy, its suboptimality gap scales with its pathlength and improvement factor. We further show, that for non-smooth nonlinear systems \acrshort*{ediss} is implied by exponential stability under certain conditions. The assumptions are verified on the suboptimal linear quadratic MPC use case and on a numerical example. The generality of the provided analysis enables the study of other examples where the suboptimality is due to unknown system parameters or cost functions, for example in the fields of adaptive and online control. We leave the analysis of these use cases as well as the treatment of measurement and process noise to future work.

%% file: SA_appendix.tex
\begin{appendix}
\subsection*{Proof of Lemma \ref*{lem:delta_V_bound}}
\begin{proof}
Given a state $\xi\in \mathcal{D}$ and some $N\in \mathbb{R}_+$, let
\begin{equation*}
    V(k_0,\xi) = \sum_{k=0}^{N-1}\phi(k+k_0,k_0,\xi)^\top \phi(k+k_0,k_0,\xi).
\end{equation*}
Then 
\begin{align*}
    V(k_0&,\xi)=\\
     &\xi^\top \xi + \sum_{k=1}^{N-1}\phi(k+k_0,k_0,\xi)^\top \phi(k+k_0,k_0,\xi) \geq \|\xi\|^2,
\end{align*}
and, from Definition \ref{def:exponential_stability}
\begin{equation*}
    V(k_0,\xi) \leq \sum_{k=0}^{N-1}d^2\|\xi\|^2\lambda^{2k}\leq \frac{d^2}{1-\lambda^{2}}\|\xi\|^2.
\end{equation*}
Thus, \eqref{eq:converse_lyapunov_1} is satisfied with $c_1=1$ and $c_2 = \frac{d^2}{1-\lambda^{2}}$. To show that \eqref{eq:converse_lyapunov_2} holds, consider
\begin{align*}
    &V(k_0 +1, f(k_0,\xi)) - V(k_0,\xi)\\
    &= \sum_{k=0}^{N-1} \left(\|\phi(k+k_0+1,k_0,\xi)\|^2 - \|\phi(k+k_0,k_0,\xi)\|^2\right)\\
    &= \|\phi(N+k_0,k_0,\xi)\|^2 - \|\xi\|^2 \leq d^2\lambda^{2N}\|\xi\|^2 - \|\xi\|^2\\
    &=- \left(1-d^2\lambda^{2N}\right)\|\xi\|^2.
\end{align*}
Choosing $N$ large enough such that  $d^2\lambda^{2N}<1$, ensures \eqref{eq:converse_lyapunov_2} holds with $\beta^2 = 1 - \frac{\left(1-d^2\lambda^{2N}\right)}{c_2}\in (0,1)$ since $c_2\geq c_1 =1$ and $c_2 \in \mathbb{R}_+$. Finally, for some $(\xi,\zeta)\in \mathcal{D} \times \mathcal{D}$ and $k_0 \in \mathbb{N}$, denote $\Delta \phi(k,k_0,\xi,\zeta):= \phi(k+k_0,k_0,\xi)-\phi(k+k_0,k_0,\zeta)$ and consider
\begin{align*}
    &|V(k_0,\xi) - V(k_0,\zeta)| \\
    &=\left|\sum_{k=0}^{N-1} \left(\|\phi(k+k_0,k_0,\xi)\|^2 - \|\phi(k+k_0,k_0,\zeta)\|^2\right)\right| \\
    \begin{split}
    &\leq \sum_{k=0}^{N-1}\left(\|\phi(k+k_0,k_0,\xi)\|+\|\phi(k+k_0,k_0,\zeta)\|\right)\\
    &\qquad \qquad \qquad \qquad \qquad \qquad \qquad\cdot\|\Delta \phi(k,k_0,\xi,\zeta)\|
    \end{split}\\
    &\leq \sum_{k=0}^{N-1}d\lambda^k\left(\|\xi\|+\|\zeta\|\right)\cdot L_f^k\|\xi-\zeta\|\\
    &= \left(\|\xi\|+\|\zeta\|\right)\|\xi-\zeta\| \sum_{k=0}^{N-1}d\lambda^k L_f^k,
\end{align*}
where the last inequality follows from the  exponential stability and $L_f$-Lipschitz continuity of the nonlinear mapping. Taking $c_3 =  \sum_{k=0}^{N-1}d\lambda^k L_f^k$ completes the proof.
\end{proof}

\subsection*{Proof of Theorem \ref{the:demidovich}}
\begin{proof}
    The condition in \eqref{eq:demidovich} implies the contraction of the dynamics within the contraction region $\mathcal{\tilde{D}}_0$, as defined in \cite{lohmiller1998contraction}. We will first show that there exists a $r_{\mathcal{\tilde{D}}_0}$  such that for all initial states $x \in \mathcal{B}(r_{\mathcal{\tilde{D}}_0})$ the dynamics \eqref{eq:nonlinearsystem} evolve in a forward invariant set contained in $\mathcal{\tilde{D}}_0$. To see this, for some $r \in \mathbb{R}_+$, and $k \in \mathbb{N}_{\geq k_0}$ define $
    \mathcal{E}_0(k,r):= \{x\in \mathbb{R}^n \ | \ \|x\|_{P(k)} \leq r\}$, and $r^\star:= \max\{r>0 \ |  \mathcal{E}_0(k,r) \subseteq \mathcal{\tilde{D}}_0, ,\;\forall k \in \mathbb{N}_{\geq k_0}\}$. This implies that if a state satisfies $x\in \mathcal{E}_0(k,r^\star)$, then  $x\in \mathcal{\tilde{D}}_0$ necessarily. The largest ball contained in the intersection of the sets $\mathcal{E}_0(k,r^\star),\; k \in \mathbb{N}_{\geq k_0}$ is then given by $\mathcal{B}(\frac{r^\star}{\bar{P}})$, where $P(k)\preceq \bar{P}I$ for all $k\in \mathbb{N}_{\geq k_0}$, and $\bar{P}>0$ exists since $P(k)$ is uniformly bounded. Consider now the differential displacement  $\delta x$ for the dynamics \eqref{eq:nonlinearsystem}, defining an infinitesimal displacement at a given time $k$. The shortest path integral with respect to the metric $P(k)$ between $0$ and a given state $x_k$ at time $k$ is given by
    \begin{equation*}
            V_{\ell}(\delta x_k, k):= \int_0^{x_k}\|\delta x\|_{P(k)} = \|x_k\|_{P(k)},
    \end{equation*}
    where the second equality follows from the fact that the shortest path integral corresponds to the Riemannian distance (see e.g.  \cite{singh2017robust,tran2018convergence, tsukamoto2021contraction}). Using the same arguments as in \cite[Thm 2.8]{tsukamoto2021contraction} or \cite[Thm. 15]{tran2018convergence}, one can then show that for any $x_k \in \mathcal{B}(\frac{r^\star}{\bar{P}})$, it holds that $V_{\ell}(\delta x_{k+1}, k+1) \leq \rho V_{\ell}(\delta x_k, k) < r^\star$. This, in turn, implies that $\|x_{k+1}\|_{P(k+1)} < r^\star \implies x_{k+1} \in \mathcal{E}_0(k+1, r^\star) \implies x_{k+1} \in \mathcal{\tilde{D}}_0 $. Taking $r_{\mathcal{\tilde{D}}_0} = \frac{r^\star}{\bar{P}}$ completes the claim that the dynamics are forward invariant within $\mathcal{\tilde{D}}_0$. For any initial state in $\mathcal{B}(r_{\mathcal{\tilde{D}}_0})$, uniform local exponential incremental stability with a rate $\rho$ then follows from standard Lyapunov  arguments, see e.g.\cite[Thm 2.8]{tsukamoto2021contraction}. 
\end{proof}

\subsection*{Proof of Theorem \ref*{the:converse_edelta_s}}
The proof hinges on extending results from \cite{jiang2002converse},\cite{sontag1996new}, \cite{angeli2002lyapunov} and \cite{tran2018convergence}. Before presenting the proof, we provide the following auxiliary definition and theorem. 

\begin{definition} \label{def:local_set_exponential_stability}
\cite{sontag1996new,jiang2002converse} Given a closed, positively invariant set $\mathcal{A}\subset \mathbb{R}^n$, the system \eqref{eq:nonlinearsystem} is uniformly exponentially stable in $\mathcal{D}$ with respect to $\mathcal{A}$, if  there exist constants $d \in \mathbb{R}_+$ and $\lambda\in (0,1)$, such that for all $\xi \in \mathcal{D}$ and $k \in \mathbb{N}_{\geq k_0}$
\begin{equation*}
    |\phi(k,k_0,\xi)|_{\mathcal{A}} \leq  d|\xi|_{\mathcal{A}} \lambda^{k-k_0}.
\end{equation*}
\end{definition}
The extension of the converse Lyapunov results from \cite{angeli2002lyapunov} and \cite{jiang2002converse} to the  case of stability in a forward invariant region is included below for completeness.
\begin{theorem}
    \label{the:set_lyapunov_function}
    If the system \eqref{eq:nonlinearsystem} is uniformly  exponentially stable in $\mathcal{D}$ with respect to some closed set $\mathcal{A}$, then there exists a function $V: \mathbb{N}_{\geq k_0}\times \mathcal{D}\rightarrow \mathbb{R}_+$ and constants $c_1,c_2\in \mathbb{R}_+$ and $\beta \in (0,1)$, such that for all $\xi \in \mathcal{D}$, and $k \in \mathbb{N}_{\geq k_0}$
    \begin{align}
        \label{eq:set_lyapunov_function_1}
        c_1 |\xi|^2_{\mathcal{A}} \leq V(k,\xi) \leq c_2 |\xi|^2_{\mathcal{A}},\\
        \label{eq:set_lyapunov_function_2}
        V\left(k+1,f(k,\xi)\right) \leq \beta^2 V(k,\xi).
    \end{align}
\end{theorem}
\begin{proof}
    Following the same line of arguments from the proof of Lemma \ref*{lem:delta_V_bound}, consider a $x\in \mathcal{D}$ and some $N\in \mathbb{R}_+$ and let
    \begin{equation}
        \label{eq:lyapunov_set}
        V(k_0,\xi) = \sum_{k=0}^{N-1}|\phi(k+k_0,k_0,\xi)|^2_{\mathcal{A}},
    \end{equation}
    for some $k_0 \in \mathbb{N}$. Then,
    \begin{equation*}
        V(k_0,\xi) = |\xi|^2_{\mathcal{A}} + \sum_{k=1}^{N-1}|\phi(k+k_0,k_0,\xi)|^2_{\mathcal{A}} \geq |\xi|^2_{\mathcal{A}},
    \end{equation*}
    and, from Definition \ref*{def:local_set_exponential_stability}
    $$
    V(k_0,\xi) \leq \frac{d^2}{1-\lambda^2}|\xi|^2_{\mathcal{A}}.$$
    Next, consider
    \begin{align*}
        &V(k_0+1,f(k_0,\xi)) - V(k_0,\xi)\\
         &= \sum_{k=0}^{N-1} \left(|\phi(k_0+k+1,k_0,\xi)|^2_{\mathcal{A}} - |\phi(k_0+k,k_0,\xi)|^2_{\mathcal{A}}\right)\\
        &= |\phi(k_0+N,\xi)|^2_{\mathcal{A}} - |\xi|^2_{\mathcal{A}} \leq d^2\lambda^{2N}|\xi|^2_{\mathcal{A}} - |\xi|^2_{\mathcal{A}}\\
        &= - \left(1-d^2\lambda^{2N}\right)|\xi|^2_{\mathcal{A}}.
    \end{align*}
    Choosing $N$ large enough such that $d^2\lambda^{2N}<1$ completes the proof with $c_1 =1$, $c_2  = \frac{d^2}{1-\lambda^2}$, and $\beta^2 = 1 - \frac{\left(1-d^2\lambda^{2N}\right)}{c_2}\in (0,1)$ since $c_2\geq c_1 =1$ and $c_2 \in \mathbb{R}_+$.
\end{proof}

\begin{proof}[Proof of Theorem \ref{the:converse_edelta_s}]
    Consider the augmented system
    \begin{equation*}
        \left\{\begin{array}{l}
            \xi_{k+1}=f(k,\xi_k) \\
            \zeta_{k+1}=f(k,\zeta_k).
            \end{array}\right.
    \end{equation*}
We define the diagonal as the set $\Delta : = \{\left[\xi^\top , \xi^\top\right]^\top \in \mathcal{D} \times \mathcal{D}: \xi \in \mathcal{D}\}$. Let $z:= [\xi^\top, \zeta^\top]^\top \in \mathcal{D} \times \mathcal{D}$, then it is shown in \cite{angeli2002lyapunov} that 
\begin{equation}
    \label{eq:set_equivalence}
    |z|_{\Delta} = \frac{\sqrt{2}}{2} \|\xi-\zeta\|.
\end{equation}
Then, considering the evolution of the combined system
\begin{equation}
    \label{eq:combined_incremental_system}
    z_{k+1} = \mathcal{F}(k,z_k): = \begin{bmatrix}
        f(k,\xi_k)\\
        f(k,\zeta_k)
    \end{bmatrix},
\end{equation}
 one can note that $z_k \in \mathcal{D}\times \mathcal{D}$ for all $k\in \mathbb{N}_{\geq k_0}$, for any $(\xi_{k_0}, \zeta_{k_0}) \in \mathcal{D} \times \mathcal{D}$, since $\mathcal{D}$ is forward invariant. It follows that system \eqref{eq:nonlinearsystem} is exponentially incrementally stable in $\mathcal{D}$ if and only if \eqref{eq:combined_incremental_system} is  exponentially stable in $\mathcal{D}$ with respect to the diagonal $\Delta$. 
 
 Moreover, using Theorem \ref*{the:set_lyapunov_function}, the combined system \eqref{eq:combined_incremental_system} admits a Lyapunov function satisfying \eqref{eq:set_lyapunov_function_1}-\eqref{eq:set_lyapunov_function_2}. Given a $z\in \mathcal{D}\times \mathcal{D}$, it follows from the equivalence in \eqref{eq:set_equivalence}, and the proof of Theorem \ref*{the:set_lyapunov_function}, that \eqref{eq:incremental_lyapunov_1}-\eqref{eq:incremental_lyapunov_2} are satisfied with $c_1 =\frac{\sqrt{2}}{2}$, and $c_2  = \frac{d^2}{\sqrt{2}-\lambda^2\sqrt{2}}$ for the following Lyapunov function
 \begin{equation*}
    V(k_0, z) = \sum_{k=0}^{N-1}|\phi_\mathcal{F}(k_0+k,k_0,z)|^2_{\Delta},
\end{equation*}
where $\phi_\mathcal{F}:\mathbb{N}_{\geq k_0} \times\mathbb{N}_{\geq k_0}\times \mathbb{R}^{2n} \rightarrow \mathbb{R}^{2n}$ is the state transition matrix for the system \eqref{eq:combined_incremental_system}.

We show the inequality \eqref{eq:lyapunov_difference_incremental} next. For any $(\xi,\zeta,\tilde{\xi},\tilde{\zeta}) \in \mathcal{D}\times \mathcal{D}\times \mathcal{D}\times \mathcal{D}$, $k \in \mathbb{N}_{\geq k_0}$, $w: = [\tilde{\xi}^\top, \tilde{\zeta}^\top]^\top$ and $z$ defined as above
\begin{align*}
    \begin{split}
        &|V(k_0,z) - V(k_0,w)| \\
        &=\left|\sum_{k=0}^{N-1} \left(|\phi_\mathcal{F}(k_0+k,k_0,z)|^2_{\Delta} - |\phi_\mathcal{F}(k_0+k,k_0,w)|^2_{\Delta}\right)\right|
    \end{split}
    \\
    \begin{split}
        &\leq\sum_{k=0}^{N-1}\left(\|\phi_\mathcal{F}(k_0+k,k_0,z) -\phi_\mathcal{F}(k_0+k,k_0,w)\|\right) \\
        &\qquad \qquad \cdot \left(|\phi_\mathcal{F}(k_0+k,k_0,z)|_{\Delta} + |\phi_\mathcal{F}(k_0+k,k_0,w)|_{\Delta}\right)
    \end{split}
    \\
    \begin{split}
        &\leq \sum_{k=0}^{N-1}\frac{d\lambda^k}{\sqrt{2}}\left\Vert\begin{bmatrix}
        \phi(k_0+k,k_0,\xi) - \phi(k_0+k,k_0,\tilde{\xi}) \\
        \phi(k_0+k,k_0,\zeta) - \phi(k_0+k,k_0,\tilde{\zeta})
    \end{bmatrix}\right\Vert\\
    &\qquad \qquad \qquad \qquad \qquad \qquad \qquad \cdot\left(\|\xi-\zeta\|+\|\tilde{\xi}-\tilde{\zeta}\|\right)
    \end{split}
    \\
    \begin{split}
    &\leq \left(\|\xi-\tilde{\xi}\| + \|\zeta-\tilde{\zeta}\|\right)\left(\|\xi-\zeta\|+\|\tilde{\xi}-\tilde{\zeta}\|\right)\sum_{k=0}^{N-1}\frac{d\lambda^k L^k}{\sqrt{2}},
    \end{split}
\end{align*}
where the last two inequalities follow from the Lipschitz continuity and exponential stability (with respect to $\Delta$) of the solutions. The first inequality follows from  the following. For any two vectors $\xi,\zeta\in \mathcal{D}$ and closed, convex set $\mathcal{A}$
\begin{align*}
    |\xi|_\mathcal{A} - |\zeta|_{\mathcal{A}} &= |\xi|_{\mathcal{A}} - \|\zeta-\zeta_p\|\leq \|\xi-\zeta_p\| - \|\zeta-\zeta_p\|\\
    &\leq |\|\xi-\zeta_p\| - \|\zeta-\zeta_p\|| \leq \|\xi-\zeta\|,
\end{align*}
where $\zeta_p \in \mathcal{A}$ is such that $|\zeta|_{\mathcal{A}} = \|\zeta-\zeta_p\|$, which exists and is unique, since $\mathcal{A}$ is closed and convex. Finally, choosing $c_3 = \sum_{k=0}^{N-1}\frac{d\lambda^k L^k}{\sqrt{2}}$ completes the proof.
\end{proof}

\subsection*{Proof of Theorem \ref{prop:eis-eiiss}}
\begin{proof}
    Given the nonlinear system \eqref{eq:nonlinearsystem}, consider the evolution of two, respectively unperturbed and perturbed trajectories
    \begin{align*}
        x_{k+1} &= f(k,x_k),\\
        y_{k+1} &= f(k,y_k) +  w_k,
    \end{align*} 
    for some $(x_{k_0}, y_{k_0}) \in \mathcal{D} \times \mathcal{D}$, where $w_k \in  \mathcal{B}(r_w)$, for some $r_w \in \mathbb{R}_+$ is such that $y_k\in \mathcal{D}$ for all $k\in \mathbb{N}_{\geq k_0}$. Given $f$ is uniformly exponentially incrementally stable in $\mathcal{D}$, then from Theorem \ref*{the:converse_edelta_s}, there exists a Lyapunov function $V:\mathbb{N}_{\geq k_0} \times \mathcal{D}\times \mathcal{D}\rightarrow \mathbb{R}_+$ satisfying \eqref{eq:incremental_lyapunov_1}-\eqref{eq:lyapunov_difference_incremental}. Then, for any $(x,y) \in \mathcal{D} \times \mathcal{D}$, $k\in \mathbb{N}_{\geq k_0}$, and $w\in \mathcal{B}(r_w)$
    \begin{align*}
        \begin{split}
            &V\left(k+1,f(k,x),f(k,y)+w\right) - V\left(k,x,y\right)=
        \end{split}
        \\
        \begin{split}
            &=V\left(k+1,f(k,x),f(k,y)\right) - V\left(k,x,y\right)\\
             &+ V\left(k+1,f(k,x),f(k,y)+w\right)- V\left(k+1,f(k,x),f(k,y)\right)
        \end{split}
        \\
        \begin{split}
            &\leq - \left(1- \beta^2\right) c_1 \|x-y\|^2\\
             &\quad + c_3\|w\|\left(\|f(k,x)-f(k,y)-w\| + \|f(k,x)-f(k,y)\|\right)\\
            &\leq - \left(1- \beta^2\right) c_1 \|x-y\|^2 + c_3\|w\|^2 + 2c_3L_f \|w\|\|x-y\|, 
        \end{split}
    \end{align*}
    where the first inequality follows from Theorem \ref*{the:converse_edelta_s}, and the second from the triangle inequality. Completing the square, and denoting $c_4 := (1-\beta^2)c_1$ it follows from above that
    \begin{align*}
        \begin{split}
            &V\left(k+1,f(k,x),f(k,y)+ w\right) - V\left(k,x,y\right) \leq
        \end{split}
        \\
        \begin{split}
        &\leq \frac{-3c_4}{2}\|x-y\|^2\\
        &\qquad \qquad + \left(\frac{\sqrt{c_4}}{\sqrt{2}}\|x-y\| + \frac{c_3L_f\sqrt{2}\|w\|}{\sqrt{c_4}}\right)^2 - \frac{2c_3^2L_f^2 \|w\|^2}{c_4}
        \end{split}
        \\
        \begin{split}
        &\leq -\frac{c_4}{2}\|x-y\|^2 + \frac{2c_3^2L_f^2 \|w\|^2}{c_4^2}\\
        &\leq -\frac{c_4}{2c_2}V(k,x,y)+ \frac{2c_3^2L_f^2 \|w\|^2}{c_4},
        \end{split} 
    \end{align*}
    where the second inequality follows by the fact that $\left(a +b\right)^2\leq 2a^2 + 2b^2$ for any $a,b \in \mathbb{R}$, and the last inequality from Theorem \ref*{the:converse_edelta_s}. Finally, rearranging the Lyapunov equations it follows that
    \begin{equation*}
        V\left(k+1,f(k,x),f(k,y)+w\right) \leq \rho^2 V\left(k,x,y\right) + c_5 \|w\|^2,
    \end{equation*}
    for all $k\in \mathbb{N}_{\geq k_0}$, with $\rho^2:= 1 - \frac{c_4}{2c_2}\in (0,1)$, since $c_2>c_4$, and $c_5:= \frac{2c_3^2L_f^2}{c_4}$. Unrolling the recursion and using \eqref{eq:incremental_lyapunov_1}
    \begin{align*}
        &c_1\|x_k-y_k\|^2\\
         &\leq \rho^{2(k-k_0)} V(k_0,x_{k_0},y_{k_0}) + c_5\sum_{i=k_0}^{k-1}\rho^{2(k-i-1)} \|w_i\|^2 \\
        &\leq c_2 \rho^{2(k-k_0)} \|x_{k_0}-y_{k_0}\|^2 +c_5\sum_{i=k_0}^{k-1}\rho^{2(k-i-1)} \|w_i\|^2.
    \end{align*}
    Dividing by $c_1$ and taking the square root, completes the proof
    \begin{align*}
        \|x_k-y_k\| \leq \sqrt{\frac{c_2}{c_1}} \rho^{(k-k_0)} \|x_{k_0}&-y_{k_0}\|\\
        & +\sqrt{\frac{c_5}{c_1}}\sum_{i=k_0}^{k-1}\rho^{k-i-1} \|w_i\|.
    \end{align*}
\end{proof}
\end{appendix}